\documentclass[a4paper,onecolumn,11pt,accepted=2025-12-04]{quantumarticle}
\pdfoutput=1
\usepackage{graphicx} 
\usepackage{subfig}
\usepackage{geometry}
\usepackage[utf8]{inputenc}
\usepackage[T1]{fontenc}
\usepackage{setspace}
\usepackage{amsmath}
\usepackage{bbm}
\usepackage{float}
\usepackage{microtype}
\usepackage{amssymb}
\usepackage{amsthm}
\usepackage{physics}
\newtheorem{mydef}{Definition}

\usepackage{tikz}
\usepackage{hyperref}
\usepackage{pgf-pie}
\usepackage{graphicx}
\usepackage{pdfpages}
\usepackage{color}
\usepackage{mathtools}
\usepackage{enumerate}
\usepackage[inkscapelatex=false]{svg}
\usepackage[OT1]{fontenc}
\usepackage[
    backend=biber,
    style=nature]{biblatex} 
\usepackage{authblk}

\newcommand{\Okl}{\Omega_{k,l}}

\newcommand{\E}{\mathcal{E}}

\newcommand{\F}{\mathcal{F}}
\newcommand{\C}{\mathcal{C}}
\newcommand{\mP}{\mathcal{P}}
\newcommand{\mQ}{\mathcal{Q}}

\newcommand{\Zd}{\mathcal{Z}_d}
\newcommand{\Z}{\mathcal{Z}}
\newcommand{\w}{\omega}
\newcommand{\denop}{\mathcal{D}}
\newcommand{\hs}{\mathcal{H}}
\DeclareMathOperator{\myTr}{Tr}
\newcommand{\otAB}{\overset{AB}{\otimes}}

\newcommand{\id}{\mathbbm{1}}
\newcommand{\Oe}{| \Omega(e) \rangle}
\newcommand{\Pe}{| \Omega(e) \rangle \langle \Omega(e) |}
\newcommand{\prob}{\mathbbm{p}}
\newcommand{\Prob}{\mathbbm{P}}
\newcommand{\Onull}{| \Omega(\Vec{0}, \Vec{0}) \rangle}
\newcommand{\Pnull}{| \Omega_{0, 0} \rangle \langle \Omega_{0,0}| }
\newtheorem{theorem}{Theorem}
\newtheorem{corollary}[theorem]{Corollary}
\newtheorem{lemma}[theorem]{Lemma}
\newtheorem{prop}[theorem]{Proposition}

\title{A Novel Stabilizer-based Entanglement Distillation Protocol for Qudits}
\author{Christopher Popp}
\email{christopher.popp@univie.ac.at}
\author{Tobias C. Sutter}
\email{tobias.christoph.sutter@univie.ac.at}
\author{Beatrix C. Hiesmayr}
\email{beatrix.hiesmayr@univie.ac.at}
\affiliation{University of Vienna, Faculty of Physics, Währingerstrasse 17, 1090 Vienna.}

\begin{document}
\maketitle
\onehalfspacing
\begin{abstract}
Entanglement distillation, the process of converting weakly entangled states into maximally entangled ones using Local Operations and Classical Communication (LOCC), is pivotal for robust entanglement-assisted quantum information processing in error-prone environments. A construction based on stabilizer codes offers an effective method for designing such protocols.
By analytically investigating the effective action of stabilizer protocols for systems of prime dimension $d$, we establish a standard form for the output states of recurrent stabilizer-based distillation. This links the properties of input states, stabilizers, and encodings to the properties of the protocol.
Based on those insights, we present a novel two-copy distillation protocol, applicable to all bipartite states in prime dimension, that maximizes the fidelity increase per iteration for Bell-diagonal states. The power of this framework and the protocol is demonstrated through numerical investigations, which provide evidence for superior performance in terms of efficiency and distillability of low-fidelity states compared to other well-established recurrence protocols.
By elucidating the interplay between states, errors, and protocols, our contribution advances the systematic development of highly effective distillation protocols, enhancing our understanding of distillability.
\end{abstract}
\section{Introduction}
Quantum information science offers exciting potentials to quantum technology like superior computational power, secure communication and improved sensing by leveraging non-classical resources. One of the most valuable and intriguing quantum resources is entanglement.
Prominent applications using this resource include quantum teleportation \cite{bennett_teleporting_1993}, quantum dense coding \cite{bennett_communication_1992, werner_all_2001}, or measurement-based quantum computing \cite{briegel_measurement-based_2009}.
An important class of systems, especially in the context of quantum communication, are shared systems between two parties, named Alice and Bob. Assuming that each party possesses $d$-level systems, called qudits, entanglement is shared in the form of joint quantum states that cannot be locally described by the individual systems, but only if the combined global system is considered. The gold standard of this resource are two-qudit maximally entangled states, so-called Bell states.
Due to interactions with the environment, this resource is generally not available in its pure form, but affected by noise. Depending on the level of noise, the entanglement of the shared pair is effectively reduced or even destroyed. 
One method to deal with this problem is entanglement distillation. The objective for Alice and Bob is to use local operations on their individual systems and classical communication (LOCC) to transform several pairs of noisy and therefore weakly entangled states to a smaller number of strongly entangled states. Not all entangled states can be distilled via LOCC, due to the existence of bound entanglement of states that are positive under partial transposition and the potential existence of undistillable states with negative partial transposition \cite{horodecki_mixed-state_1998, hiesmayr_bipartite_2025}. \\
Entanglement distillation (sometimes also called purification) was first introduced for bipartite two-level systems, i.e., qubits, and then developed to become more efficient \cite{bennett_purification_1996, bennett_concentrating_1996, deutsch_quantum_1996, yan_measurement-based_2022,zhou_purification_2020} or allow for the distillation of entangled qudits \cite{alber_efficient_2001, horodecki_reduction_1999, vollbrecht_efficient_2003}. Two main classes of distillation protocols can be identified. Recurrence protocols operate on a fixed set of input pairs iteratively, while hashing or breeding protocols operate on the whole ensemble of pairs. While the later class has in principle a higher efficiency, i.e., the inverse expected number of input states required to produce one highly entangled pair, they require low levels of noise. Recurrence protocols, on the other hand, suffer from lower efficiency but can operate on states with stronger noise. For both classes, many distillation protocols have been shown to be special cases of two generalizing schemes. Permutation-based schemes \cite{dehaene_local_2003} use permutations of products of Bell states that can be realized by LOCC. Stabilizer-based protocols utilize stabilizer codes \cite{gottesman_stabilizer_1997, ashikhmin_nonbinary_2001}, i.e., codes  based on a commutative subgroup of the Pauli group of errors that take the inherent structure of error processes into account. In this contribution, we mainly consider recurrence protocols based on the stabilizer scheme.\\
Reflecting a general connection between error correction and entanglement distillation \cite{wilde_quantum_2008, dur_entanglement_2003, dur_entanglement_2007}, it was shown that for any stabilizer code an entanglement distillation protocol can be defined \cite{matsumoto_conversion_2003}. In these recurrence-type protocols, which have also been extended to breeding-type protocols \cite{matsumoto_breeding_2024}, the two parties carry out stabilizer measurements that project their local states of several qudits to subspaces called codespaces. A basis of eigenstates of these codespaces are the codewords, and a corresponding basis transformation is named encoding. It was shown that for $d=2$ and a special class of encodings, the choice of encoding affects the performance of the corresponding protocol \cite{watanabe_improvement_2006}. However, in a general setting, there is no systematic way to derive powerful protocols regarding their efficiency and minimal fidelity requirements yet. In Ref.\cite{miguel-ramiro_efficient_2018}, a recurrence protocol was introduced that includes information about the input state to choose between two distillation routines that can be related to certain stabilizer protocols. Showing a good efficiency for high-fidelity states, but failing to distill low-fidelity states, the information about the input state is not effectively leveraged to derive the optimal protocol.\\
In this contribution, we develop a general theory of stabilizer-based distillation in prime dimension that allows to relate information about the input states, error operators and used stabilizer codes to the efficacy and properties of the corresponding stabilizer distillation protocol. We demonstrate the power of this method by proposing a new distillation protocol that exhibits superior performance compared to other recurrence protocols for several state families. The paper is organized as follows. In Section \ref{sec:stab-dist-prot}, after introducing the notation of Bell states, error operators and the stabilizer formalism, we summarize the standard routine of stabilizer-based distillation. Section \ref{sec:stabCosedDist} analyzes the effect of errors in a given stabilizer encoding to derive a standard form, relating adjustable parameters of the protocols and information of the input states to properties of the output states. In Section \ref{sec:two_copy_distillation_prime}, the developed theory is applied to the case of two-copy distillation in prime dimensions to propose the distillation protocol FIMAX that is shown to maximize the fidelity increase of Bell-diagonal states in each iteration for all stabilizer protocols. The efficacy of FIMAX is demonstrated by comparing it to other prominent recurrence protocols, where it shows notable results regarding efficiency and especially the distillation of low-fidelity states. Finally, the results are discussed in Section \ref{sec:discussion_and_conclusion}.
\section{Preliminaries for stabilizer-based entanglement distillation}
\label{sec:stab-dist-prot}
\subsection{Bell states and stabilizer codes}
\begin{figure}[ht]
    \centering
    \vspace{-1em}
    \includegraphics[width=0.4\linewidth]{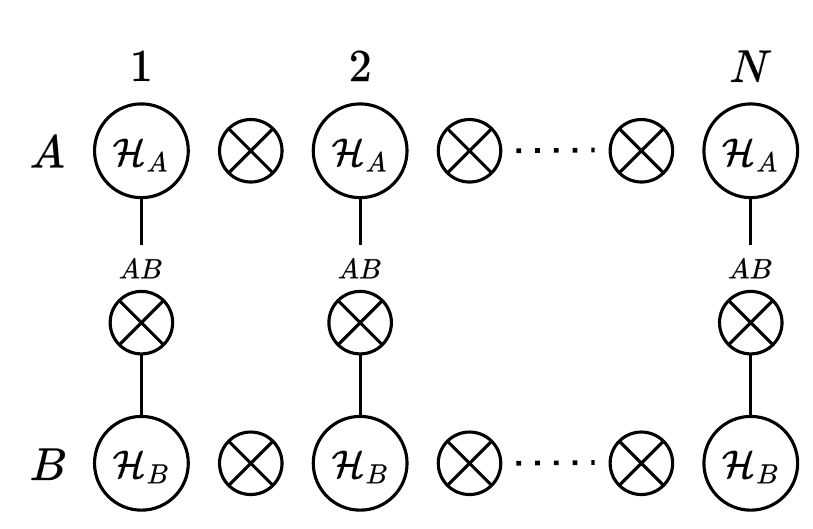}
    \caption{Hilbert space $\hs^{\otimes N}$ of $N$-copies of a bipartite system $\hs = \hs_A \otAB \hs_B$.}
    \label{fig:hsTensor}
\end{figure}
\noindent Let $\hs \equiv \hs_A \otAB \hs_B \cong \mathbbm{C}^d \otimes \mathbbm{C}^d$ be the Hilbert space of bipartite quantum states with dimension of the subsystems $d$ and $\hs^{\otimes N} \equiv \bigotimes_{n=1}^N \hs \cong \hs_A^{\otimes N} \otAB \hs_B^{\otimes N}$ be the corresponding Hilbert space of $N$-copy quantum states. In this work, we restrict $d$ to be a prime number, simplifying the analytical investigations. The tensor structure of this Hilbert space is depicted in Figure \ref{fig:hsTensor}. Let $\denop(\hs)$ be the set of density operators on $\hs$.
We write $\Zd \equiv \mathbbm{Z}/d\mathbbm{Z}$ for the quotient ring of integers with addition and multiplication modulo $d$. Note that $\Zd$ is a field for prime $d$ as each nonzero element has a unique inverse. Let $\w \equiv \exp(\frac{2 \pi i}{d})$, and denote equality up to a phase by $\propto$. Complex conjugation $(\star)$ is defined in the computational basis.\\
In the following, we define the relevant objects for this work related to stabilizers of the group of Weyl-Heisenberg errors. Assuming that the reader is familiar with the basics of the stabilizer formalism, we do not provide a complete introduction. For more information about stabilizers in arbitrary dimension, see Refs.~\cite{knill_non-binary_1996, rains_nonbinary_1999}. In Appendix A, we provide an example for the introduced objects.
\begin{mydef}[Weyl(-Heisenberg) operators, Weyl(-Heisenberg) errors]\label{def:weylOpsErrs}
    \begin{flalign}
        W_{k,l} &:= \sum_{j=0}^{d-1}\w^{j k} |j\rangle \langle j+l|,~~k,l \in \Zd \\
        \E_N &:= \left\{  W(e) ~|~ e = (\Vec{k}, \Vec{l}) \in \Zd^N \times \Zd^N \right\} := \left\{ \bigotimes_{n = 1}^N W_{k_n,l_n}  ~|~ k_n,l_n \in \Zd \right\}
    \end{flalign}
\end{mydef}
\noindent 
The Weyl-Heisenberg operators satisfy the Weyl relations, i.e.,
\begin{flalign}
    \label{eq:weylRelations}
    \begin{aligned}
           W_{k_1,l_1}W_{k_2,l_2} &= \w^{l_1 k_2}~W_{k_1+k_2, l_1+l_2},  \\
           W_{k,l}^\dagger &= \w^{k l}~W_{-k, -l} = W_{k,l}^{-1},
    \end{aligned}
\end{flalign}
implying that the set of Weyl-errors $\E_N$ forms a group under multiplication, if we identify errors that are equal up to a phase, i.e., $W_{k_1,l_1}W_{k_2,l_2} \equiv W_{k_1+k_2, l_1+l_2}$. $E \in \E_N$ are called \emph{error operators}. For $E = W(e)$, $e \equiv (\Vec{k}, \Vec{l}) \in \Zd^N \times \Zd^N$ are called \emph{error elements}. The group structure of $\E_N$ induces a group structure via the Weyl relations \eqref{eq:weylRelations} for error elements on $\Zd^N \times \Zd^N$ with addition modulo $d$.
\begin{mydef}[Bell states]
\begin{flalign}
  \label{bellStates}
     |\Okl\rangle &:= (W_{k,l} \otAB \mathbbm{1}_d)|\Omega_{0,0}\rangle := (W_{k,l} \otAB \mathbbm{1}_d)\frac{1}{\sqrt{d}}\sum_{i} | i \rangle \otAB |i \rangle,~~i,k,l \in \Zd \\
    \Oe &:= (W(e)\otAB \id_d^{\otimes N}) |\Omega_{0,0}\rangle^{\otimes N} = \bigotimes_{n=1}^N |\Omega_{k_n,l_n} \rangle, ~~e \equiv (\Vec{k}, \Vec{l}) \in \Zd^N \times \Zd^N 
\end{flalign}
\end{mydef}
\begin{mydef}[Stabilizer group, generating operators, generating elements] \ \newline
    A stabilizer group or stabilizer $S$  is an abelian subgroup of $\E_N$. If $\lbrace W(g_1), \dots, W(g_p) \rbrace$ forms a minimal generating set, each $W(g_j)$ is called a generating operator and we write $S = \langle W(g_1), \dots, W(g_p) \rangle$. Each $g_j \in 
    \Zd^N \times \Zd^N$ is called a generating element. Given one choice of generating elements, the corresponding subgroup of $\Zd^N \times \Z_d^N$, i.e., $G_S := \lbrace g \in \Zd^N \times \Zd^N ~|~ W(g) \in S \rbrace$ is also denoted as $G_S = \langle g_1, \dots, g_p \rangle$.
\end{mydef}
\noindent
Let $S$ be a stabilizer and $\lbrace g_1,\dots,g_p \rbrace$ be the generating elements. The spectra of the generators have the following property:
\begin{lemma}
\label{thm:spec_gen}
For prime dimension $d$, each generator $W(g) \neq \mathbbm{1}_d^{\otimes N}$ has $d$ distinct eigenvalues with equal multiplicity.
\end{lemma}
\begin{proof}
First consider the special case of $d=2$. The eigenvalues of a Weyl operator $W_{k,l}$ are $\lbrace \w_y = \w^{y-\frac{1}{2}kl} | y \in \Z_2 \rbrace $ (c.f.  Lemma \ref{thm:errors_eigenstates}.1), with $y \in \Z_2$. Note that there exists exactly one eigenvalue for each $y$, so $y$ is uniformly distributed in $\Z_2$. The eigenvalues of a generator $W(g) =  \bigotimes_{n=1}^N W_{k_n,l_n}$ are, consequently,
$\left\{ \prod_{n=1}^N \w_{y_n} = \w^{\sum_n (y_n - \frac{1}{2}k_nl_n)} | y_n \in \Z_2 \right\} $. Since the additional phase $\sum_n \frac{1}{2}k_nl_n$ is fixed for a given generator, each eigenvalue corresponds to an $x:=\sum_n y_n \in \Z_2$, which is again uniformly distributed in $\Z_2$ because each $y_n$ is uniformly distributed and addition in $\Z_2$ is a bijection. This implies equal multiplicity of the two eigenvalues. 
A similar argument holds for prime dimensions $d\geq 3$. In this case, each $W_{k,l} \neq W_{0,0}$ has $d$ distinct eigenvalues $\lbrace \w_{y}=\w^{y} ~|~ y \in \Zd \rbrace$  (c.f. Lemma \ref{thm:errors_eigenstates}.1). Again, the non-degeneracy of each eigenvalue of $W_{k,l}$ implies that the corresponding values $y$ are uniformly distributed over $\Zd$. The eigenvalues of $W(g) =  \bigotimes_{n=1}^N W_{k_n,l_n}$ are $\left\{ \prod_{n=1}^N \w_{y_n} = \w^{\sum_n y_n} | y_n \in \Zd \right\} $, for which each eigenvalue can again be associated with a unique element in $\Zd$: $x :=  \sum_n y_n \in \Z_d$. The distribution of $x$ in $\Zd$, given the uniform distribution of $y_n$, is again uniform due to the bijective addition in $\Zd$, so each of the $d$ eigenvalues has the same multiplicity.
Note that any trivial Weyl operator $W_{0,0}$ contained in $W(g)$ does not introduce new eigenvalues to the spectrum, but the multiplicity of each eigenvalue is increased by a factor of $d$. Consequently, each non-trivial generator $W(g)$ for prime dimension $d$ has precisely $d$ eigenvalues with equal multiplicity. 
\end{proof}
\noindent While this proof only depends on properties of the Weyl operators, we note that Lemma \ref{thm:spec_gen} also follows directly as a special case from the more general Theorem 2 in Ref.~\cite{ashikhmin_nonbinary_2001}. \\
Using this property, we can associate each of the $d$ eigenspaces of a generator $W(g_j)$ with a number $x_j\in \Zd$. Let $x=(x_1,\dots,x_p) \in \Zd^p$ and $\mQ(x)\subset \hs_A^{\otimes N}$ be the joint eigenspace of $\lbrace W(g_j) \rbrace_{j=1}^p$, such that for all $|\phi\rangle \in \mQ(x)$ and $j \in \lbrace 1, \dots, p\rbrace$, we have $W(g_j)|\phi\rangle = \w_{x_j}|\phi\rangle$ with $x_j \in \Zd$ as in Lemma \ref{thm:spec_gen}. 
As an example, consider $d=2$, $N=2$ and the generators $W(g_1) = \begin{pmatrix} \begin{smallmatrix} 0 & 1 \\  -1 & 0 \end{smallmatrix} \end{pmatrix} \otimes \mathbbm{1}_2, 
W(g_2) = \begin{pmatrix} \begin{smallmatrix} 1 & 0 \\  0 & -1 \end{smallmatrix} \end{pmatrix} \otimes \mathbbm{1}_2
$. The eigenvalues of $W(g_1)$ are $\lbrace -i, i \rbrace$. We associate $x_1 \leftrightarrow \w^{x_1 -1/2}$, so that $x_1 = 0 \leftrightarrow \w^{-\frac{1}{2}} = -i$ and $x_1 = 1 \leftrightarrow \w^{\frac{1}{2}}= i$. 
The eigenvalues of $W(g_2)$ are $\lbrace -1, 1 \rbrace$ and we associate $x_2 \leftrightarrow \w^{x_2}$, implying $x_2 = 0 \leftrightarrow \w^0 = 1$ and $x_2 = 1 \leftrightarrow \w^1 = -1$. Therefore, the eigenspace $\mQ((0,1))$ corresponds to the pair of eigenvalues $(-i, -1)$.\\
Decomposing the Hilbert space as $\hs_A^{\otimes N} \cong \bigoplus_{x\in \Zd^p} \mQ(x)$ defines the so-called \emph{codespaces} $\mQ(x)$ in $\hs_A^{\otimes N}$. Given $p \leq N$ generators of a stabilizer in $\E_N$ for prime dimension $d$, the codespaces have dimension $\dim(\mQ(x)) = d^{N-p}~ \forall x \in \Zd^p$. This has been shown in Theorem 2 of Ref.~\cite{ashikhmin_nonbinary_2001} and can also be seen from the following argument.
 Each codespace $\mQ(x)$ corresponds to a $p$-tuple $x\ = (x_1,\dots, x_p)  \in \Zd^p$, where each entry $x_j$ corresponds to a specific eigenvalue of $W(g_j)$. Lemma \ref{thm:spec_gen} above states that each generator has exactly $d$ eigenvalues, so there are $d^p$ distinct $p$-tuples and corresponding codespaces. Since also the multiplicity of each eigenvalue is equal, all codespaces must have equal dimension, which implies $\dim(\mQ(x)) = \dim(\hs_A^{\otimes N})/d^p = d^{N-p}.$\\
Let $S^{\star}$ be the stabilizer with complex conjugated elements of $S$ and generating operators $W^{\star}(g_j)$. $\mQ^{\star}(x)\subset \hs_B^{\otimes N}$ is the joint eigenspace of $\lbrace W^{\star}(g_j) \rbrace_{j=1}^p$, such that for all $|\phi\rangle \in \mQ^{\star}(x)$ and $j \in \lbrace 1, \dots, p\rbrace$, we have $W^{\star}(g_j)|\phi\rangle =\w^\star_{x_j}|\phi\rangle$.
\begin{mydef}[Error coset] \label{def:error_coset}\ \\
Let $S \subset \E_N$ be a stabilizer and let $G_S \subset \Zd^N \times \Zd^N$ be the corresponding subgroup of error elements, such that $S = \lbrace W(g) ~|~g \in G_S \rbrace$. 
Given an error element $e \in \Zd^N \times \Zd^N$, the error coset is defined as $C(e):= e + G_S = \lbrace e + h ~|~ h\in G_S \rbrace$.
\end{mydef}
\noindent Note that a coset $C(e)$ only depends on the error element $e$ and on the stabilizer $S$ via the corresponding subgroup $G_S$ and not on a particular choice of stabilizer generators/generating elements.
\begin{mydef}[Codeword] \label{def:codeword} \ \\
    Let $\hs_A^{\otimes N} \cong \bigoplus_{x\in \Zd^p} \mQ(x)$ be decomposed into $d^{N-p}$-dimensional codespaces of a stabilizer $S$. 
    Let $\lbrace |u_{x,k} \rangle \rbrace_{x \in \Zd^p, ~k \in \Zd^{N-p}}$ be an orthonormal basis of $\hs_A^{\otimes N}$ with $|u_{x,k} \rangle \in \mQ(x)$, i.e., $W(g_j)|u_{x,k} \rangle = w_{x_j}|u_{x,k}\rangle, ~ \forall k,j$. The vectors $| u_{x,k} \rangle \in \hs_A^{\otimes N}$ are called codewords of $S$ in $\hs_A^{\otimes N}$. The codewords of $S^{\star}$ in $\hs_B^{\otimes N}$ are denoted by $| u^{\star}_{x,k} \rangle \in \hs_B^{\otimes N}$.
\end{mydef}
\noindent Since all codespaces are of equal dimension $d^{N-p}$ for prime $d$, we can define a simple mapping from the computational basis to the basis of codewords. This mapping is called encoding.
\begin{mydef}[Encoding] \label{def:encoding} \ \\
     Let $\lbrace |x \rangle \rbrace_{x \in \Zd^p}$ be the computational basis of $\hs_A^{\otimes p}$ and $\lbrace |k \rangle \rbrace_{k \in \Zd^{N-p}}$ be the computational basis of $\hs_A^{\otimes N-p}$. 
    A unitary operator $U$ on $\hs_A^{\otimes N} \cong \hs_A^{\otimes p} \otimes \hs_A^{\otimes N-p}$ is an encoding for a stabilizer $S$ in $\hs_A^{\otimes N}$ if $~\forall x \in \Zd^p,k\in \Zd^{N-p}:~ U (|x\rangle \otimes |k\rangle) =: |u_{x,k}\rangle$ is a codeword of $S$ in $\mQ(x) \subset \hs_A^{\otimes N}$.
\end{mydef} 
\noindent
If $U:|x\rangle \otimes |k\rangle \mapsto |u_{x,k}\rangle \in \hs_A^{\otimes N}$ is an encoding for $S$ in $\hs_A^{\otimes N}$, then $U^{\star}:|x\rangle \otimes |k\rangle \mapsto |u^{\star}_{x,k}\rangle \in \hs_B^{\otimes N}$ is an encoding for $S^{\star}$ in $\hs_B^{\otimes N}$.
Let $\mP(x): \hs_A^{\otimes N} \rightarrow \mQ(x)$ be the \emph{projection to the codespace} $\mQ(x)$. One has $\mP(x) = \sum_k | u_{x,k} \rangle \langle u_{x,k} |$ for any basis of codewords $|u_{x,k} \rangle$.
The same holds for projections corresponding to $S^{\star}$: $\mP^{\star}(x) = \sum_k | u^{\star}_{x,k} \rangle \langle u^{\star}_{x,k} |$. For any encoding $U:|x\rangle \otimes |k\rangle\mapsto |u_{x,k}\rangle$, we have $|\Omega(\Vec{0},\Vec{0}) \rangle = \frac{1}{d^{N/2}}\sum_{x\in \Zd^{p}}\sum_{k\in \Zd^{N-p}} |u_{x,k}\rangle \otAB |u^{\star}_{x,k}\rangle$ \cite{matsumoto_conversion_2003}.
\begin{mydef}[Symplectic product decomposition] \label{def:symprod} \ \\    
    The symplectic product of two error elements, $e=(\Vec{k},\Vec{l})$ and $f=(\Vec{m}, \Vec{n})$ is defined as $\langle e,f \rangle := \sum_{i = 0}^{N-1} l_i m_i - k_i n_i$.
    It induces a decomposition of the set of errors according to its values $s=(s_1, \dots , s_p)$ with respect to a stabilizer with generating elements $\lbrace g_1, \dots, g_p\rbrace: \E_N \cong \bigoplus_s \E(s)$, with 
    \begin{flalign}
        \E(s) := \lbrace e \in \Zd^N \times \Zd^N ~|~ \langle  g_j, e \rangle = s_j ~\forall j=1,\dots,p \rbrace.
    \end{flalign}
\end{mydef} 

\subsection{The standard stabilizer distillation protocol}
A bipartite $[N,K]$ distillation protocol transforms $N$ copies of a bipartite \emph{input state} $\rho_{in} \in \denop (\hs^{\otimes N})$ into $K$ copies of a highly entangled bipartite \emph{output state} $\rho_{out} \in \denop (\hs^{\otimes K})$ via LOCC. Recurrence protocols are iteratively applied to several copies of the input state until the output state is considered close to a specified \emph{target state}. In this work, we consider the maximally entangled state $|\Omega_{0,0}\rangle$ \eqref{bellStates} as the target state for recurrence protocols based on stabilizer measurements. \\
Let $N \geq 2$ and $S$ be a stabilizer of $\E_N$ with generating elements $g_j,~j \in \lbrace 1,\dots,p\rbrace$. Assume that one party, Alice, acting locally on $\hs_A$, can perform ``stabilizer measurements'' of local observables with the same eigenspaces as $W(g_j)$. Measurement of those observables corresponds to the projection onto a codespace $\mQ(a)$. Further, assume that the second party, Bob, acting locally on $\hs_B$, can perform stabilizer measurements corresponding to projection onto the eigenspace $\mQ^\star(b)$ of $S^\star$. $a, b \in \Zd^p$ are called ``measurement outcomes". Consider a stabilizer with $p$ distinct generating elements. The main steps of a stabilizer-based $[N, N-p]$ distillation protocol are (cf.~\cite{matsumoto_conversion_2003} for details): \\ 
\ \\
\textbf{Standard stabilizer distillation protocol}: 
\begin{enumerate}   
    \item Alice and Bob perform local stabilizer measurements with outcomes $a, b$.
        \item Bob sends Alice his measurement outcome $b$. Alice may declare failure of the protocol depending on $a$ and $b$.
    \item Alice and Bob apply the inverse of stabilizer encodings $U$ (Alice) and $U^\star$ (Bob) on $\hs_A^{\otimes N}$ and $\hs_B^{\otimes N}$.
    \item Alice and Bob discard $p$ qudits that are determined by the measurement, and Alice identifies and applies a local correction operation to the remaining $N-p$ qudits.
\end{enumerate}
Note that by choosing a unitarily equivalent correction operation, the last two steps can in principle be carried out in changed sequence. Iteratively applying this protocol defines a recurrence distillation scheme \cite{matsumoto_conversion_2003}.
\section{Generalized stabilizer distillation} \label{sec:stabCosedDist}
In this section, we analyze the impact of a stabilizer protocol on a state. Exploiting the algebraic properties of the Weyl operators and related Bell states, errors, and stabilizers, the action of such a protocol can be written in a form, in which the effect of different choices regarding stabilizer, encoding, measurement, and correction operation become clearly visible. These insights allow for proving certain properties and optimizations of stabilizer-based protocols.
\subsection{The action of Weyl errors on codewords} \label{sec:errors_on_codewords}
A general input state can be written in the Bell basis as $\rho_{in} = \sum_{e,f \in \E_N} \rho(e, f) ~|\Omega(e)\rangle \langle \Omega(f)|$, with two error elements $e, f \in \Zd^N \times \Zd^N$ and the density matrix elements $\rho(e,f)$. Following the same arguments as in Ref.~\cite{matsumoto_conversion_2003}, the combined effect of the measurements with outcomes $a$ and $b$ and the application of the inverse encoding operations is of the following form after the third step of the protocol:
\begin{flalign}
    \label{eq:out_invencoding_general}
    \begin{aligned}    
    \rho_{in} &\mapsto \frac{1}{\textrm{Prob}(a=b+s)} ~(U^{-1} \otAB (U^\star)^{-1}) ~(\mP(a) \otAB \mP^\star(b)) ~\rho_{in}~  (\mP(a) \otAB \mP^\star(b))~ (U \otAB U^\star) \\
    &=\frac{d^{-(N-p)}}{\textrm{Prob}(a=b+s)}\sum_{e,f\in \E(s)} \sum_{j,l \in \Zd^{N-p}} \rho(e,f) ~(
    U^{\dagger} W(e) U~(|b\rangle\langle b| \otimes |j\rangle\langle l|)~U^{\dagger} W(f)^{\dagger} U
    ) 
    \otAB 
    (|b\rangle\langle b| \otimes |j\rangle\langle l|).
    \end{aligned}
\end{flalign}
Here, $\textrm{Prob}(a=b+s)$ denotes the probability for obtaining the measurement outcomes $a$ and $b$ with $s\equiv a-b$.\\
An important class of input states are Bell-diagonal states (BDS) (cf., e.g., \cite{baumgartner_special_2007, popp_comparing_2023}), arising naturally when local errors affect the maximally entangled state $|\Omega_{0,0}\rangle$. Let $\prob: \E_N \rightarrow \mathbbm{R}$ be a discrete probability distribution on the set of Weyl-Heisenberg errors $\E_N$ and $\Prob$ be the corresponding probability measure. For an error element $e \in \Zd^N \times \Zd^N$, $\prob(e)$ is called \emph{error probability}. If $N$-copies of maximally entangled states are affected by $e$ with probability $\prob(e)$, we can write the state as $\rho_{in} = \sum_{e\in \E} \prob(e) \Pe$. Assuming such a multi-copy Bell-diagonal input state, the state is transformed by the protocol as follows:
\begin{flalign}
    \label{eq:out_invencoding}
    \rho_{in} \mapsto \frac{d^{-(N-p)}}{\Prob(\E(s))}\sum_{e\in \E(s)} \sum_{j,l \in \Zd^{N-p}} \prob(e) ~(
    U^{\dagger} W(e) U~(|b\rangle\langle b| \otimes |j\rangle\langle l|)~U^{\dagger} W(e)^{\dagger} U
    ) 
    \otAB 
    (|b\rangle\langle b| \otimes |j\rangle\langle l|).
\end{flalign}
\noindent From \eqref{eq:out_invencoding_general} and \eqref{eq:out_invencoding} it is clear that the effect of errors $e$ in the encoding, i.e., $U^\dagger W(e) U$, is relevant for the performance of the protocol. In Ref.~\cite{watanabe_improvement_2006}, a special subset of encodings for the case $d=2$ was analyzed regarding their impact on the performance of stabilizer-based distillation. Here, we investigate all encodings for prime dimension $d$.\\ 
Starting with Lemma $\ref{thm:svalue}$, we show that an error $W(e)$ maps a codeword of the codespace $\mQ(x)$ to a generally different codespace $\mQ(x+s)$. $s$ depends solely on the generating elements of the chosen stabilizer and the acting error.
\begin{lemma}
\label{thm:svalue}
Let $S$ be a stabilizer, $x \in \Zd^p$ and $| \phi \rangle \in \mQ(x)\subset \hs_A^{\otimes N}$ be an eigenvector in the common eigenspace of generating operators $W(g_j)$ with corresponding eigenvalue $\w_{x_j}$ for $j \in \lbrace 1,\cdots,p \rbrace$. Let $W(e)$ be an error operator,  $s_j = \langle  g_j, e \rangle$, and $s = (s_1,\cdots, s_p) \in \Zd^p$. 
Then $W(e) |\phi \rangle \in \mQ(x+s)$.
\end{lemma}
\begin{proof} Using the Weyl relations \eqref{eq:weylRelations}, one has $\forall j$:
\begin{flalign*}
        W(g_j) W(e) |\phi\rangle 
        = \w^{\langle g_j, e \rangle} W(e) W(g_j) |\phi\rangle 
        = \w^{s_j} \w_{x_j} W(e) |\phi\rangle 
        = \w_{x_j+s_j} W(e)|\phi\rangle ~~\forall j~    
        \implies W(e)|\phi \rangle \in \mQ(x+s)
\end{flalign*}
\end{proof} 
\noindent Lemma \ref{thm:encoding_errors} demonstrates that for each error $e$ the precise \emph{action of the error $e$} within the codespace is determined by a $d^{N-p}$-- dimensional unitary transformation $T_{x}^{U,e}$ that depends on the codespace, the encoding and the error. In essence, these operators reflect which effect an error has on the codewords of a stabilizer code.
In Section \ref{sec:stand_form}, it is demonstrated that, together with the stabilizer measurements, these action operators determine the output state of the stabilizer protocol. 
\begin{lemma} \label{thm:encoding_errors}
Let $U$ be an encoding of a stabilizer $S$ with generating elements $\lbrace g_1,\dots,g_p\rbrace$. For each codespace $\mQ(x) \subset \hs_A^{\otimes N}$  and for each $W(e) \in \E_N$ with $(\langle g_1, e \rangle,\dots,\langle g_p, e \rangle) = (s_1,\dots,s_p)$, there exist unitary ``action'' operators $T^{U,e}_{x}: \hs_A^{\otimes N-p}\rightarrow \hs_A^{\otimes N-p}$ satisfying
    $U^\dagger W(e) U = \sum_{x\in \Zd^p} |x+s \rangle \langle x | \otimes T_{x+s}^{U,e}$.
\end{lemma}
\begin{proof}
\noindent Lemma \ref{thm:svalue} implies $W(e)|u_{b,j}\rangle \in \mQ(b+s)$, which is spanned by $\lbrace |u_{b+s,k}\rangle \rbrace_{k \in \Zd^{N-p}}$. Consequently, 
\begin{flalign*}
    U^{\dagger}W(e)U(|b\rangle\otimes|j\rangle) 
    = U^{\dagger} W(e) |u_{b,j}\rangle
    = U^{\dagger}\sum_{k\in \Zd^{N-p}}t_{b+s,k}^{e,j} |u_{b+s,k}\rangle 
    = |b+s\rangle \otimes \sum_{k\in \Zd^{N-p}}t_{b+s,k}^{e,j} |k\rangle.
\end{flalign*}
The elements $\lbrace t^{e,j}_{b+s,k}\rbrace_{k,j  \in \Zd^{N-p}}$ define the action operator $T_{b+s}^{U,e}$ of the error $e$ in the computational basis via 
\begin{flalign}
        T_{b+s}^{U,e} := \sum_{k,j\in \Zd^{N-p}}t_{b+s,k}^{e,j} |k\rangle\langle j|.
\end{flalign}
One then has $U^{\dagger}W(e)U(|b\rangle\otimes|j\rangle) = |b+s\rangle \otimes T_{b+s}^{U,e}~|j\rangle$, implying $U^\dagger W(e) U = \sum_{x\in \Zd^p} |x+s \rangle \langle x | \otimes T_{x+s}^{U,e}$.
\end{proof}
\noindent The following two lemmas show that the group properties of errors and stabilizers are also reflected by the action operators $T$ and therefore by the effective action of errors in a given encoding. This will be leveraged to derive a simple form of the output state of a stabilizer protocol in Section \ref{sec:stand_form}. \\
Lemma \ref{thm:lin_structure_error_action} illustrates that the linear structure of $\E_N$ naturally extends to the action of errors.
\begin{lemma}
    \label{thm:lin_structure_error_action}
    Let $e, f$ be two error elements and $T^{U, e}_{x+s_e}$ and $T^{U, f}_{x+s_f}$  be the corresponding actions in the codespaces. We then have:
    \begin{flalign}
        \label{eq:lin_struct}
        U^\dagger W(e+f) U &\propto \sum_x |x+s_e+s_f \rangle \langle x| \otimes T^{U,e}_{x+s_e} T^{U,f}_{x+s_f} \\
        U^\dagger W(-e) U &\propto \sum_x |x-s_e \rangle \langle x| \otimes (T^{U,e}_x)^\dagger
    \end{flalign}
\end{lemma}
\begin{proof}
Follows directly from $W(e+f) \propto W(e)W(f)$ and $W(-e) \propto W(e)^\dagger$ and Lemma \ref{thm:encoding_errors}.
\end{proof}
\noindent The following lemma shows that the actions of two errors that are related by a generating element are equivalent up to a phase. This directly implies that errors of the same coset (Definition \ref{def:error_coset}) are also equivalent up to a phase.
\begin{lemma}
    \label{thm:coset_error_action}
    Let $U$ be an encoding for a stabilizer $S$ as in Lemma \ref{thm:encoding_errors} and let $T_{x+s}^{U,e}$ be the corresponding action for an error element $e$. The following equality holds:
    \begin{flalign}
        T_{x+s}^{U,e+g_j} = \w_{x_j} T_{x+s}^{U,e}~~\forall j \in \lbrace 1,\dots, p \rbrace.
    \end{flalign}
\end{lemma}
\begin{proof}
Noting $U^\dagger W(g_j) U = \sum_{x\in \Zd^p} |x \rangle \langle x | \otimes \w_{x_j} \mathbbm{1}_{d^{N-p}}$, the claim follows from Lemma \ref{thm:lin_structure_error_action} \eqref{eq:lin_struct} with $f = g_j$.
\end{proof}
\noindent Proposition \ref{thm:gen_stab_encs} establishes a connection between all possible encodings. Given an encoding, all other encodings are related by concatenation via a block-diagonal unitary transformation. Conversely, any such concatenation provides another encoding. Moreover, the action of errors in a concatenated encoding can be directly derived.
\begin{prop}
    \label{thm:gen_stab_encs}
    Let $U$ be an encoding for a stabilizer $S$ with generating elements $\lbrace g_1,\dots,g_p\rbrace$. 
    \begin{enumerate}[(i)]
        \item A unitary $V$ is another encoding for $S$ if and only if for each codespace $\mQ(x)$, $x\in \Zd^p$, there exists unitary $(\dim(\mQ(x)\cross \dim(\mQ(x))$-matrices $Y_x$ so that $ V = U~(\sum_x |x \rangle\langle x|\otimes Y_x)$.
        \item Let $V = U~(\sum_x |x \rangle\langle x|\otimes Y_x)$. If $
            ~U^\dagger W(e) U = \sum_x |x+s \rangle \langle x | \otimes T_{x+s}^{U,e}~
        $ 
        as in Lemma~\ref{thm:encoding_errors}. Then $V^\dagger W(e) V = \sum_x |x+s \rangle \langle x | \otimes Y_{x+s}^\dagger T_{x+s}^{U,e} Y_x$.
    \end{enumerate}
\end{prop}
\begin{proof} \ \newline
    \vspace{-2em}
    \begin{enumerate}[(i)]
        \item Let $U$ and $V$ be encodings such that $U~(|y\rangle \otimes |l \rangle) \equiv |u_{y,l}\rangle$ and $V ~(|x\rangle \otimes |k \rangle) \equiv |v_{x,k}\rangle$.
        Then
        \begin{flalign*}
            (\langle y | \otimes \langle l|)~U^\dagger V ~(|x\rangle \otimes |k \rangle) = \langle u_{y,l}|v_{x,k}\rangle = \delta_{x,y}\langle u_{y,l}|v_{x,k}\rangle
        \end{flalign*}
        and consequently
        \begin{flalign*}
         U^\dagger V = \sum_{x\in \Zd^p} |x\rangle\langle x| \otimes 
            \sum_{l,k \in \Zd^{N-p}} \langle u_{x,l}|v_{x,k}\rangle ~|l \rangle\langle k| =:  \sum_x |x\rangle\langle x| \otimes  Y_x.                         
        \end{flalign*}
        Conversely, assume $ V = U~(\sum_x |x \rangle\langle x|\otimes Y_x)$. We need to show that $V ~ (|x\rangle \otimes |k\rangle)$ defines a codeword, i.e., that $W(g_j) V ~(|x\rangle \otimes |k\rangle) = \w_{x_j} V ~(|x\rangle \otimes |k\rangle) ~~\forall j \in \lbrace 1,\dots, p \rbrace ,~x \in \Zd^p, ~k \in \Zd^{N-p}$. This can be shown by direct calculation:
        \begin{flalign*}
            W(g_j) V~ (|x\rangle \otimes |k\rangle )
            &= W(g_j) ~U ~(|x\rangle \otimes Y_x|k\rangle) 
            = W(g_j) ~U ~(|x\rangle \otimes \sum_l y_{x,l}|l\rangle) \\
            &= \w_{x_j} \sum_l y_{x,l} U ~(|x\rangle \otimes |l\rangle) 
            = \w_{x_j} U ~(|x\rangle \otimes Y_x|k\rangle) \\  
            &= \w_{x_j} V ~(|x\rangle \otimes |k\rangle).            
        \end{flalign*}
        \item Using Lemma \ref{thm:encoding_errors} for both $U$ and $V$ yields
        \begin{flalign*}
            V^\dagger W(e) V &= (\sum_z |z \rangle\langle z|\otimes Y^\dagger_z)~U^\dagger
            ~W(e)~
            U~(\sum_y |y \rangle\langle y|\otimes Y_y) \\
            &= (\sum_z |z \rangle\langle z|\otimes Y^\dagger_z) 
            (\sum_x |x+s \rangle \langle x | \otimes T_{x+s}^{U,e})
            (\sum_y |y \rangle\langle y|\otimes Y_y)  \\
            &= \sum_x |x+s \rangle \langle x | \otimes Y^\dagger_{x+s}T_{x+s}^{U,e} Y_x.
        \end{flalign*}
    \end{enumerate}
\end{proof}
\subsection{Standard form of stabilizer distillation protocols}
\label{sec:stand_form}
The results of the previous section suggest that the effect of a stabilizer protocol (cf. \eqref{eq:out_invencoding_general} \eqref{eq:out_invencoding}) can be made more concise by considering the action of errors in a given encoding. First, Proposition \ref{thm:standard_form_general} provides a simplified form of the effect of the stabilizer measurements and decoding operations. Theorem \ref{thm:standard-form} introduces a ``standard form'' of the stabilizer distillation protocol for Bell-diagonal input states, incorporating the group properties of the Weyl errors.
\begin{prop}
\label{thm:standard_form_general}
    Let $S$ be a stabilizer with $p$ generating elements defining the symplectic partition of errors in $\E_N \cong \bigoplus_s \E(s)$ and $U$ be an encoding. 
    Let $\rho_{in} = \sum_{e,f \in \E_N} \rho(e, f) ~|\Omega(e)\rangle \langle \Omega(f)|$ be a general input state in the Bell basis.
    Let $\textrm{Prob}(a=b+s)$ be the probability of obtaining such outcomes for the stabilizer measurements. After projection of the state onto the codespace $\mQ(a) \otAB \mQ^\star(b)$ with $s=a-b$, applying $U^{-1}\otAB(U^\star)^{-1}$ and discarding the first $p$ copies, the output state $\rho_{out} \in \bigotimes_{n=p+1}^N \hs_A \otAB \hs_B$ is
    \begin{flalign}
    \label{eq:standard_form_general}
    \rho_{out} = \frac{1}{\emph{Prob}(a=b+s)}\sum_{e,f\in \E(s)}  \rho(e,f) ~(
        T_{b+s}^{U,e} \otAB \id)~ 
        \Pnull^{\otimes N-p}
        ~(T_{b+s}^{U,f} \otAB \id)^{\dagger}.
    \end{flalign}
\end{prop}
\begin{proof}
    Starting with the state in the form of Eq.\eqref{eq:out_invencoding_general}, Lemma \ref{thm:svalue} implies that after the projective measurements with outcomes $a$ and $b$, only terms relating to errors contained in $\E(s)$ have nonzero components. Lemma \ref{thm:encoding_errors} determines the form of those copies that are not trivially determined by the measurement outcomes $a,b$ via the action of errors in the applied encoding.
\end{proof}
\noindent For Bell-diagonal input states, further simplification can be achieved by considering error cosets because, according to Lemma \ref{thm:coset_error_action}, errors from the same coset only differ by a phase that cancels for Bell-diagonal input states. We can therefore represent all action operators $T_x^{U,e}$ for errors of the same coset $C$ by a single coset action operator $T_x^{U,C}$. As a $N$-copy Bell-diagonal state corresponds to a probability distribution $\prob$ on $\E_N$, combining errors of the same coset induces a distribution for the error coset probabilities.
\begin{mydef}
\label{def:coset_defs}
\begin{flalign}
    &\C := \lbrace C~|~ C \textrm{ is an error coset} \rbrace \\
    &\C(s) := \lbrace C \in \C~|~ C \subset \E(s) \rbrace \\  
    &\Prob(C) = \sum_{e\in C} \prob(e)\\
    &T_{b+s}^{U,C} := T_{b+s}^{U,e} \textrm{ for an arbitrary }  e \in C \\
    &T_{b+s}^{U,(C_1 - C_2)} := T_{b+s}^{U,(e_1 - e_2)} \textrm{ for arbitrary } e_1 \in C_1, e_2 \in C_2
\end{flalign}     
\end{mydef} 
\noindent Given Bell-diagonal input states and combining these definitions with Proposition \ref{thm:standard_form_general} allows deriving a form of the output state that solely depends on objects relating to cosets instead of individual errors.
\begin{theorem}    \label{thm:standard-form}
    Let $\rho_{in} = \sum_{e\in \E} \prob(e) \Pe$ be a bipartite $N$-copy Bell-diagonal state, inducing a probability measure $\Prob$ on $\E_N$. Let $S$ and  $U$ be as in Proposition \ref{thm:standard_form_general}.
    After projection of the state onto the codespace $\mQ(a) \otAB \mQ^\star(b)$ with $s=a-b$, applying $U^{-1}\otAB(U^\star)^{-1}$ and discarding the first $p$ copies,  $\rho_{in}$ is transformed to  $\rho_{out}~\in~\denop((\hs_A~\otAB~\hs_B)^{\otimes N-p})$ in the so-called ``standard form":
    \begin{flalign}
    \label{eq:standard_form}
    \rho_{out} = \frac{1}{\Prob(\E(s))}\sum_{C \in \C(s)} \Prob(C)~
            (T_{b+s}^{U,C} \otAB \id)~ 
                \Pnull^{\otimes N-p}
            ~(T_{b+s}^{U,C} \otAB \id)^{\dagger}.
    \end{flalign}
\end{theorem} 
\begin{proof}
    Assuming Bell-diagonal form, we have $\textrm{Prob}(a=b+s) = \Prob(\E(s))$ and $\rho(e,e) \equiv \prob(e)$. Proposition \ref{thm:standard_form_general} implies the claimed form, by noting that Lemma \ref{thm:coset_error_action} allows identifying errors of the same coset $C$ (cf. Definition \ref{def:coset_defs}), as their action operators only differ by a phase that cancels for diagonal elements.
\end{proof}
\noindent In this standard form, the effect of the adjustable parameters of the protocol become clearly visible:
\begin{itemize}
    \item \emph{Input state $\rho_{in}$}: The input state reflects the probability measure $\Prob$ on $\E_N$ and thus the probability of errors.
    \item \emph{Stabilizer $S$}: The number of generators $p$ of $S$ determines how many copies are used to gain information from measurements and are discarded afterward. The remaining $N-p$ copies form a maximally entangled state that is affected locally by an error with some probability. $S$ determines the decomposition of the Hilbert space in codespaces and of $\E_N$ in the sets $\E(s)$ as well as in cosets $C$.
    \item \emph{Measurement and classical communication}: The measurement outcomes $a$ and $b$ effectively limit the set of possible errors to $\E(s)$. If only Alice needs to know $s$ for further operations, one-way communication from Bob to Alice is enough. Otherwise, two-way communication is required.
    \item \emph{Encoding $U$}: The encoding fixes a basis of codewords for the codespaces and therefore determines the effect of local errors according to the error action operators $T_{b+s}^{U,e}$. The output state $\rho_{out}$ can be written as mixed Bell states that are locally affected by these error actions on Alice's system.
\end{itemize}
For the sake of clarity, we omit the indices representing the dependence of the error action operators on the chosen encoding and the subspace they act on. Hence, we write $T^{U,C}_{b+s} \equiv T_C$ if there is no risk of confusion.
\subsection{Fidelities and local error correction in the standard form}
The standard form of Theorem \ref{thm:standard-form} allows finding local correction operations and calculating fidelities for a target state $|\Omega(\Vec{k}, \Vec{l})\rangle, ~(\Vec{k}, \Vec{l}) \in \Zd^{N-p} \times \Zd^{N-p}$. Note that these are $(N-p)$-copy states, as the $p$ copies containing only information about the measurement outcomes are discarded in the protocol.
Let $V \otAB \id_{d^{N-p}}$ be a local unitary, applied by Alice (w.l.o.g.) depending on the measurement outcomes and transforming the output state to
\begin{flalign}
    \rho_{out} \mapsto \frac{1}{\Prob(\E(s))}\sum_{C \in \C(s)} \Prob(C)~
            (VT_C \otAB \id)~ 
                \Pnull^{\otimes N-p}
            ~(VT_C \otAB \id)^{\dagger}.
\end{flalign}
Choosing $V = T_{\hat{C}}^\dagger$ with some coset $\hat{C}$ shows that a fidelity of $\frac{\Prob(\hat{C})}{\Prob(\E(s))}$ can be achieved:
\begin{flalign}
     \begin{aligned}
    \rho_{out} \mapsto ~&\frac{\Prob(\hat{C})}{\Prob(\E(s))}(
         \Pnull^{\otimes N-p} 
        + \sum_{C \in \C(s)\setminus\hat{C}} \Prob(C)
            (T_{C-\hat{C}} \otAB \id)~ 
                \Pnull^{\otimes N-p}~
            (T_{C-\hat{C}} \otAB \id)^{\dagger}.
    \end{aligned}
\end{flalign}
In Ref.~\cite{matsumoto_conversion_2003} it was shown that this fidelity with $|\Omega(\Vec{0}, \Vec{0})\rangle = |\Omega_{0,0}\rangle^{\otimes N-p}$ can also be obtained if Alice applies $W(e)^{-1}, e\in \hat{C}$ before the inverse encoding $U^{-1}$. The standard form \eqref{eq:standard_form} has the advantage that the fidelities for all Bell states for any encoding can be directly calculated from its error actions $T_C$. This makes a quantitative comparison of different encodings possible.\\
The fidelity between any multi-copy Bell state and the output state in standard form \eqref{eq:standard_form} is
\begin{flalign}
    \label{eq:fidelities}
    \begin{aligned}    
    \F(\Vec{k}, \Vec{l}) := \langle \Omega(\Vec{k}, \Vec{l})| \rho_{out} |\Omega(\Vec{k}, \Vec{l}) \rangle  &= \frac{1}{\Prob(\E(s))}\sum_{C \in \C(s)} \Prob(C)~
            |\langle \Omega(\Vec{k}, \Vec{l})|(T_C \otAB \id)\Onull |^2 \\
    &= \frac{1}{\Prob(\E(s))}\sum_{C \in \C(s)} \Prob(C)~ 
    |\myTr(\frac{1}{d^{N-p}}  W^\dagger(\Vec{k}, \Vec{l}) T_C)|^2 .
    \end{aligned}
\end{flalign}
For the last equality of \eqref{eq:fidelities}, the following identity for $\Onull$ of dimension $D$ and all $(D\cross D)$ matrices $M$ is used:
\begin{flalign}
    \label{eq:max_ent_trace_identity}
    \langle \Omega(\Vec{0},\Vec{0})| M \otAB \id \Onull = \frac{1}{D}  \myTr(M).
\end{flalign}
With the \emph{Weyl representation} of $T_C$,
\begin{flalign}
    \label{eq:weyl_representation}
    T_C = \sum_{(\Vec{i}, \Vec{j}) \in \Zd^{2(N-p)}} \beta_C(\Vec{i}, \Vec{j})~ W(\Vec{i},\Vec{j}), ~~~
    \beta_C(\Vec{i}, \Vec{j}) := \frac{1}{d^{N-p}} \myTr(W^\dagger(\Vec{i}, \Vec{j})~ T_C) \in \mathbbm{C},        
\end{flalign}
the fidelities are directly related to the coefficients of $\beta_C$ in this representation. More precisely, Corollary \ref{thm:fidelity_coset_contribution} demonstrates that the total fidelity is a weighted sum of so-called \emph{coset fidelities} $f_C(\Vec{k},\Vec{l}) := |\beta_C(\Vec{k}, \Vec{l})|^2$.
\begin{corollary}
    \label{thm:fidelity_coset_contribution}
    Let $\rho_{out}$ be as in Theorem \ref{thm:standard-form}. Then the fidelity $\F(\Vec{k}, \Vec{l})$ of $\rho_{out}$ can be expressed as
    \begin{flalign}
    \label{eq:coset_fidelities}
        \F(\Vec{k}, \Vec{l}) ~&=~ \frac{1}{\Prob(\E(s))}\sum_{C \in \C(s)} \Prob(C) |\beta_C(\Vec{k}, \Vec{l})|^2 ~=:~\frac{1}{\Prob(\E(s))}\sum_{C \in \C(s)} \Prob(C) f_C(\Vec{k}, \Vec{l}).
    \end{flalign}
\end{corollary}
\begin{proof}
Follows directly from \eqref{eq:fidelities} with \eqref{eq:weyl_representation}.
\end{proof} 
\noindent Proposition \ref{thm:coset_contribution_properties} states properties of the coset fidelities and of the states $\lbrace T_C \otAB \id ~|\Omega(\Vec{k},\Vec{l})\rangle\rbrace$, emerging in the standard form \eqref{eq:standard_form}. These properties will allow showing that the protocol proposed in section \ref{sec:two_copy_distillation_prime} maximizes the increase in fidelity in each iteration among all stabilizer-based protocols. We introduce the following notation for cosets $C_1, C_2$:
\begin{flalign}
    \delta_{C_1,C_2} := 
    \begin{cases}
        1 $ if $ C_1 = C_2\\
        0 $ else. $
    \end{cases}
\end{flalign}
\begin{prop}
\label{thm:coset_contribution_properties}
Let $S$ be a stabilizer, $G_S$ the set of corresponding error elements as in Definition~\ref{def:error_coset}, $T_C$ be as in Theorem \ref{thm:standard-form} and $f_C(\Vec{k},\Vec{l})$ be as in Corollary \ref{thm:fidelity_coset_contribution}. Then,
\begin{enumerate}[(i)]
    \item $\forall C: ~ \sum_{(\Vec{k},\Vec{l})} f_C(\Vec{k},\Vec{l}) = 1$.
    \item $\forall (\Vec{k}, \Vec{l})$, states in 
        $\lbrace T^\dagger_C \otAB \id ~|\Omega(\Vec{k},\Vec{l}) \rangle ~|~ C \in \C(s) \rbrace  \text{ are orthonormal up to a phase } 
        ~\Longleftrightarrow~
        \frac{1}{d^{N-p}}\myTr(T^\dagger_C) \propto \delta_{C,G_S}.$
    \item $\lbrace T^\dagger_C \otAB \id ~|\Omega(\Vec{k},\Vec{l}) \rangle ~|~ C \in \C(s) \rbrace  \text{ is an ONB of } (\hs_A\otAB\hs_B)^{\otimes{N-p}}  \implies \sum_{C\in \C} f_C(\Vec{k},\Vec{l}) = 1$.
\end{enumerate}    
\end{prop}
\begin{proof} \ \\
    \vspace{-2em}
    \begin{enumerate}[(i)]
        \item $\forall C,~ T_C$ is unitary, implying $1 = \sum_{(\Vec{k}, \Vec{l})} \beta^\star(\Vec{k},\Vec{l})\beta(\Vec{k},\Vec{l}) 
            = \sum_{(\Vec{k}, \Vec{l})} f(\Vec{k},\Vec{l})$ by \eqref{eq:weyl_representation}.
        \item Consider the identity \eqref{eq:max_ent_trace_identity}. For $C_1, C_2 \in \C(s)$ this implies
        \begin{flalign*}
            \langle \Omega(\Vec{k},\Vec{l}) | (T_{C_1} \otAB \id) (T^\dagger_{C_2} \otAB \id) | \Omega(\Vec{k},\Vec{l}) \rangle
             ~=~ \frac{1}{d^{N-p}} \myTr( T_{C_1} T^\dagger_{C_2})
             ~\propto~ \frac{1}{d^{N-p}} \myTr( T_{C_1-C_2}).
        \end{flalign*}
        The last equation follows from Lemma \ref{thm:coset_error_action} and Lemma \ref{thm:lin_structure_error_action} together with Definition \ref{def:coset_defs}. 
        Assume the right-hand side of the equivalence. The equation above implies orthonormality up to a phase for the states $T^\dagger_C \otAB \id ~|\Omega(\Vec{k},\Vec{l}) \rangle$ if $C_1-C_2 = G_S \Leftrightarrow C_1 = C_2$. Conversely, orthonormality up to a phase of elements on the left side of the equivalence implies $
            \frac{1}{d^{N-p}}\myTr(T_{C}) \propto \frac{1}{d^{N-p}}\myTr(T_{C-G_S}) 
            \propto \langle \Omega(\Vec{k},\Vec{l}) | (T_C \otAB \id) (T^\dagger_{G_S} \otAB \id) | \Omega(\Vec{k},\Vec{l}) \rangle 
            \propto \delta_{C,G_S}.
        $
        \item Comparing \eqref{eq:fidelities} with \eqref{eq:coset_fidelities} and assuming the ONB property, one has
        \begin{flalign*}
            \sum_{C\in \C} f_C(\Vec{k},\Vec{l})  &= \sum_{C \in \C}|\langle \Omega(\Vec{k}, \Vec{l})| (T_C \otAB \id)\Onull|^2 \\
            &= \sum_{C\in \C} \langle \Omega(\Vec{k}, \Vec{l})|~(T_C \otAB \id)~ \Onull \langle \Omega(\Vec{0}, \Vec{0})|~(T^\dagger_C \otAB \id) ~|\Omega(\Vec{k}, \Vec{l}) \rangle 
            = \myTr(\Onull \langle \Omega(\Vec{0}, \Vec{0}) |) = 1.
        \end{flalign*}
        \end{enumerate}
\end{proof}
\section{Two-copy distillation in prime dimension}
\label{sec:two_copy_distillation_prime}
In this section, we consider the case $N=2$ for prime dimension $d$. Introducing the canonical encoding, we apply the results of the previous section to characterize all other encodings. In section \ref{sec:fimax_protocl}, we propose a protocol that maximizes the fidelity increase in each iteration for Bell-diagonal input states and compare it numerically to other protocols in section \ref{sec:numeric_comparison}. Note that the results of the sections \ref{sec:two_copy_prime_properties} and \ref{sec:canonical_encoding} do not depend on the input state and are therefore applicable to non-Bell-diagonal states as well. 
\subsection{Error sets and stabilizers for prime dimension} \label{sec:two_copy_prime_properties}
From the standard form (cf. Theorem \ref{thm:standard-form}) it follows that the output state of the generalized stabilizer protocol is a mixed $(N-p)$-copy bipartite state. Consequently, the only number of generating elements $p$ of a stabilizer $S$, which results in a non-trivial transformation of the input state for $N=2$, is $p=1$. Therefore, all relevant stabilizer groups $S$ have exactly one generating element $g \in \Zd^2\times \Zd^2$. Since $d$ is prime, the order of each $S$ is $d$ and we can explicitly write $S = \lbrace \mathbbm{1}, W(g), W(2g), \dots ,W((d-1)g) \rbrace.$ Every cyclic group is abelian, so any error operator of $W(e) \neq \mathbbm{1}$ generates a stabilizer group and two stabilizer groups sharing one element are identical.
The following Lemma \ref{thm:partition_sizes} shows that the partitions of errors in $\E_N \cong \bigoplus_s \E(s)$ (cf. Definition \ref{def:symprod}) are of equal size for prime dimensions.
\begin{lemma}
    \label{thm:partition_sizes}
    Let $d$ be prime, $g \neq (\Vec{0},\Vec{0}) \in \Zd^N \times \Zd^N$ inducing $\E_N \cong \bigoplus_s \E(s)$. Then $\forall s \in \Zd$ we have $|\E(s)| = d^{2N-1}$.
\end{lemma}
\begin{proof}
    We define a bijective map $M: \E(0) \rightarrow \E(s)$, implying equal cardinality of the sets. Let $g = (\Vec{k}, \Vec{l}) \neq (\Vec{0}, \Vec{0}) \in \Zd^N \times \Zd^N$. Assume that there exists a component $(k_i, l_i)$ with $k_i \neq 0$ (w.l.o.g.). Let $e_0 \in \E(0)$. Such an element always exists, since $g \in \E(0)$. $d$ is prime, so $\exists k_i^{-1} \in \Zd$. Define for $e = (\Vec{m},\Vec{n})= ((m_1,\dots m_N),(n_1,\dots,n_N)) \in \E(0)$ the map $M: e_0 \mapsto e_s = ((m_1,\dots,m_i+sk_i^{-1},\dots,m_N),(n_1,\dots,n_N))$.
    $k_i^{-1}$ is unique and $sk_i^{-1} \neq 0$ for $s\neq 0$. $M$ is a bijection and $\langle g, e_s \rangle = s$, so the $d$ partitions $\E(s)$ must be of equal cardinality $d^{2N -1}$.
\end{proof}
\subsection{A canonical stabilizer encoding}
\label{sec:canonical_encoding}
In this section, a specific encoding based on the eigenvectors of the Weyl-Heisenberg operators is defined. We show that this \emph{canonical encoding} has special properties and implications for corresponding stabilizer distillation protocols.\\
First, we analyze how Weyl errors affect eigenstates of the Weyl-Heisenberg operators, and thus codewords in the canonical encoding, in Lemma \ref{thm:errors_eigenstates}. The proof relies on the technical lemmas \ref{thm:errors_eigenstates}.1 and \ref{thm:errors_eigenstates}.2 following below.
\begin{lemma}
    \label{thm:errors_eigenstates}
    Let $d$ be prime and $|\w_\lambda\rangle$ be an eigenvector of $W_{a,b}$ with eigenvalues depending on $\lambda$ as in Lemma \ref{thm:errors_eigenstates}.1. For $x,y \in \Zd$, it holds that $W_{x,y} ~|\w_\lambda\rangle = \w^{\Phi} ~|\w_{\lambda + s}\rangle$ with $\Phi = \Phi(\lambda, a, b, x, y) = t(a, b, x, y) \lambda + c(a, b, x, y)$ and $s = \langle g, e \rangle$.
\end{lemma} 
\begin{proof}
    Follows directly from the Lemmas \ref{thm:errors_eigenstates}.1 and \ref{thm:errors_eigenstates}.2.
\end{proof}
\newtheorem*{thm:eigenvecs_weyl}{Lemma \ref{thm:errors_eigenstates}.1}
\begin{thm:eigenvecs_weyl}
    \label{thm:eigenvecs_weyl}
    For prime $d>2$, the eigenvalues and eigenvectors of $W_{a,b}$, $a,b \in \Zd$ are as follows.
\begin{itemize}
    \item $a=b=0$:\\
        \textbf{Eigenvalues}: $\w_\lambda := \w^0=1, ~\lambda \in \Zd.$ \\
        \textbf{Eigenvectors}: $|\w_{\lambda}\rangle := |\lambda\rangle$
    \item $a \neq 0, b=0$: $d$ is prime, so $a^{-1}\in \Zd$ exists and is unique.\\
        \textbf{Eigenvalues}: $\w_\lambda := \w^\lambda, ~\lambda \in \Zd. $ \\
        \textbf{Eigenvectors}: $|\w_\lambda\rangle := |\lambda a^{-1} \rangle.$ 
    \item $a \in \Zd, b \neq 0:$ $d$ is prime, so $b^{-1} \in \Zd$ exists and is unique. \\
        \textbf{Eigenvalues}: $\w_\lambda := \w^\lambda, ~\lambda \in \Zd$. \\
        \textbf{Eigenvectors}: $|\w_\lambda\rangle := \frac{1}{\sqrt{d}} \sum_{j\in \Zd} 
        \w^{\Gamma_{\lambda,j}} |j\rangle,~\Gamma_{\lambda, j} := jb^{-1}\lambda - \frac{jb^{-1} (jb^{-1} -1)}{2}ab$. 
\end{itemize}
In the unique special case of $d=2$ and $a\in \mathcal{Z}_2, b \neq 0$, one has:
\begin{itemize}
    \item $a \in \mathcal{Z}_2, b \neq 0:$ \\
    \textbf{Eigenvalues}: $ \w_\lambda := \w^{\lambda-\frac{1}{2}ab}, ~\lambda \in \mathcal{Z}_2$. \\
    \textbf{Eigenvectors}: $|\w_\lambda\rangle = \frac{1}{\sqrt{2}}(|0\rangle + w_\lambda |1\rangle)$
\end{itemize}
\end{thm:eigenvecs_weyl}
\begin{proof}
    That the states $|\w_\lambda\rangle $ are eigenvectors with eigenvalues as stated can be seen  by direct calculation. Since the spectrum is non-degenerate for the nontrivial case $(a,b) \neq (0,0)$, these states form a basis.
\end{proof}
\newtheorem*{thm:errors_weyl_eigenstates2}{Lemma \ref{thm:errors_eigenstates}.2}
\begin{thm:errors_weyl_eigenstates2}
    \label{thm:errors_weyl_eigenstates}
    Let $|\w_\lambda\rangle$ be an eigenvector of $W_{a,b}$. Given $W_{x,y}$ for prime $d$, the following holds:
    \begin{itemize}
    \item $a=b=0$: $W_{x,y} ~|\w_\lambda \rangle = \w^{\Phi(\lambda,0,0,x,y)} ~|\w_{\lambda-y}\rangle :=  \w^{x(\lambda-y)} ~|\w_{\lambda-y}\rangle.$
    \item $a \neq 0, b=0$: $W_{x,y} ~|\w_\lambda\rangle = \w^{\Phi(\lambda,a, 0, x, y)} ~|\w_{\lambda -ay}\rangle :=  \w^{xa^{-1}(\lambda - ay)} ~|\w_{\lambda - ay}\rangle.$
\end{itemize}
For $d>2$:
\begin{itemize}
    \item $a \in \Zd, b \neq 0$: $W_{x,y} ~|\w_\lambda\rangle = \w^{\Phi(\lambda,a,b,x,y)} ~|\w_{\lambda + bx -ay}\rangle :=  \w^{yb^{-1}(\lambda - \frac{1}{2}a(y-b)}) ~|\w_{\lambda +bx - ay}\rangle$.
\end{itemize}
For $d=2:$
\begin{itemize}
    \item $a \in \mathcal{Z}_2, b\neq 0$: $W_{x,y}~|\w_\lambda\rangle = \w^{\Phi(\lambda,a,b,x,y)} ~|\w_{\lambda + bx - ay}\rangle :=  \w^{y(\lambda - \frac{1}{2}ab}) ~|\w_{\lambda + bx - ay}\rangle$ 
\end{itemize}
\end{thm:errors_weyl_eigenstates2}
\begin{proof}
    Direct calculation with Lemma \ref{thm:errors_eigenstates}.1.
\end{proof} 

\noindent Given a stabilizer, the product basis of eigenstates of the Weyl-Heisenberg operators contained in the generating operator $W(g) = W_{a_1, b_1} \otimes W_{a_2, b_2}$ defines a valid encoding:
\begin{mydef}(Canonical encoding)     \label{def:canonical_enc}\\
Given a stabilizer $S$ with the generating element $g$ and $W(g)=W_{a_1,b_1}\otimes W_{a_2,b_2} \in \E_2$ for prime $d$. Let $\lbrace |\w^1_\lambda\rangle\rbrace_{\lambda \in \Zd}$ be a basis of eigenvectors of $W_{a_1,b_1}$ and $\lbrace |\w^2_\lambda\rangle\rbrace_{\lambda \in \Zd}$ be a basis of eigenvectors of $W_{a_2,b_2}$ as in Lemma \ref{thm:errors_eigenstates}.1. \\
    The canonical encoding is defined as
    \begin{flalign*}
        (i)&~ U_c : |\lambda\rangle \otimes |k\rangle \mapsto |u_{\lambda,k}\rangle = |\w^1_k\rangle \otimes |\w^2_{\lambda-k}\rangle ~~~(\text{if }(a_1, b_1), (a_2,b_2) \neq (0,0)). \\
        (ii)&~U_c : |\lambda\rangle \otimes |k\rangle \mapsto |u_{\lambda,k}\rangle = |\w^1_\lambda\rangle \otimes |\w^2_k\rangle ~~~~(\text{if (w.l.o.g.) } (a_2,b_2) = (0,0)).
    \end{flalign*}
\end{mydef}
\noindent 
The canonical encoding has the interesting property that its actions of errors are Weyl-Heisenberg operators. This implies that the canonical encoding maps the Pauli group of Weyl-Heisenberg errors (\ref{def:weylOpsErrs}) onto itself, i.e., the encoding operator $U_c$ are elements of the Clifford group. Lemma \ref{thm:canonical_error_actions} identifies the precise Weyl-Heisenberg operator that represents the action of any error (coset). An example for a specific stabilizer can be found in the appendix A.
\begin{lemma}
    \label{thm:canonical_error_actions}
    Let $g=((a_1,b_1),(a_2,b_2))$ be the generator of the stabilizer and $U_c$ be the canonical encoding. Let $T_e$ be the corresponding action of errors $e = ((x_1, y_1),(x_2,y_2))$, and $T_C$ the corresponding coset action. Then $T_e \propto T_C = W_{t, -s_1}$
    with $t = t(a_1, a_2, b_1, b_2, x_1, x_2, y_1, y_2)$ and $s_1 = b_1x_1 - a_1y_1$.
\end{lemma} 
\begin{proof}
Definition \ref{def:canonical_enc} and Lemma \ref{thm:errors_eigenstates} imply
$ U_c^{-1} (W_{x_1,y_1} \otimes W_{x_2,y_2}) U_c ~|\lambda \rangle \otimes |k\rangle 
= |\lambda + s_1 + s_2 \rangle \otimes \w^{\Phi_1 + \Phi_2} |k+s_1\rangle.$ 
In the case of assumption $(i)$ in Definition \ref{def:canonical_enc},  we have $\Phi_1 = \Phi_1(k, a_1, b_1, x_1, y_1)$ and  $\Phi_2 = \Phi_2(\lambda-k, a_2, b_2, x_2, y_2)$, while in the case of assumption $(ii)$ $\Phi_1 = \Phi_1(b, a_1, b_1, x_1, y_1)$ and  $\Phi_2 = \Phi_2(k, a_2, b_2, x_2, y_2)$. Using Lemma \ref{thm:errors_eigenstates}.2 one has:
\begin{enumerate}[(i)]
    \item 
    $
    \Phi_1 + \Phi_2 =  \Phi_1(k, a_1, b_1, x_1, y_1) + \Phi_2(b-k, a_2, b_2, x_2, y_2) 
    = (t_1 - t_2)k + c_1  +t_2b 
    =: t k + c.
    $
    \item 
    $
    \Phi_1 + \Phi_2 =  \Phi_1(b, a_1, b_1, x_1, y_1) + \Phi_2(k, a_2, b_2, x_2, y_2) 
    = t_1b+c_1 +t_2 k +c_2 
    =: t k + c$.
\end{enumerate}
This implies
\begin{flalign*}
U_c^{-1} (W_{x_1,y_1} \otimes W_{x_2,y_2}) U_c ~|x\rangle \otimes |k\rangle  
    = |x+s \rangle \otimes \w^c \w^{tk} |k+s_1\rangle
    = |x+s \rangle \otimes \w^{c-ts_1} W_{t, -s_1} |k\rangle.
\end{flalign*}
and comparison to Lemma \ref{thm:encoding_errors} and Definition \ref{def:coset_defs} shows the claimed property.
\end{proof}
\noindent By definition, any stabilizer contains the unity and consequently any error element in the stabilizer coset has trivial action. The following Lemma \ref{thm:canonical_error_action_kernel} shows that the stabilizer is the only coset for which that holds.
\begin{lemma}
    \label{thm:canonical_error_action_kernel}
    Let $S$ be a stabilizer with error elements $G_S$ as in Definition \ref{def:error_coset}, $U_c$ be the canonical encoding with $T_{x+s}^{U_c,e} \propto T_C = W_{t, -s_1}$. We then have the following equivalence: $C = G_S \Leftrightarrow T_C \propto W_{0,0} = \id_d.$
\end{lemma}
\begin{proof}
    Assume $W(e) \in S$. This implies $T_e \propto T_{G_S} \propto \id_d$ by Lemma \ref{thm:coset_error_action}. 
    Conversely, for $e \in C$ assume $T_e \propto W_{0,0}$, implying $s_1 = 0$ by Lemma \ref{thm:canonical_error_actions} and $\langle g,e \rangle \equiv s=s_1 + s_2 = s_2$. If $s_2=0 \Rightarrow s=0$. Lemmas \ref{thm:svalue} --\ref{thm:encoding_errors} imply
    \begin{flalign*}
     W(e) U~|b\rangle \otimes |k\rangle = W(e) |u_{b,k}\rangle \propto |u_{b+s,k}\rangle = |u_{b,k}\rangle \implies W(e) \in S \implies C = G_S.
    \end{flalign*}
    If $s_2 \neq 0$, $\exists s^{-1}_2 \neq 0$ since  $d$ is prime.
    We then have by Lemmas \ref{thm:svalue} -- \ref{thm:lin_structure_error_action}:
    \begin{flalign*}        
    W(e)^{(s_2^{-1})} ~|u_{b,k} \rangle \propto |u_{b+s_2^{-1}s_2,k} \rangle = |u_{b,k} \rangle \implies W(e)^{(s_2^{-1})} \in S  \implies W(e) \in S \implies C = G_S.
    \end{flalign*}
\end{proof}
\noindent Combining these results with the standard form of the output state \eqref{eq:standard_form}, Lemma \ref{thm:canonical_encoding_error_actions_is_basis} demonstrates that for the canonical encoding and Bell-diagonal input state, the output state is again a mixture of pure basis states. This implies, in particular, that BDS are mapped to BDS.
\begin{lemma}
    \label{thm:canonical_encoding_error_actions_is_basis}
    In the canonical encoding for $N=2$ and $d$ prime, $\forall k,l$ $\lbrace T_C \otAB \id ~|\Okl \rangle ~|~ C \in \C(s) \rbrace$ is a basis of  $\hs_A \otimes \hs_B$.
\end{lemma}
\begin{proof}
By Lemma \ref{thm:partition_sizes}, $|\C(s)| = d^2$. Orthonormality is shown by Proposition \ref{thm:coset_contribution_properties} $(ii)$, by noting that Lemma \ref{thm:canonical_error_actions} and \ref{thm:canonical_error_action_kernel} imply 
$\frac{1}{d}\myTr(T_C) \propto \frac{1}{d}\myTr(W_{t, -s_1}) = \delta_{(t,s_1), (0,0)} W_{t, -s_1} = \delta_{C, G_S}.$
\end{proof} 
\subsection{The fidelity increase maximizing distillation protocol ``FIMAX''}\label{sec:fimax_protocl}
Based on the standard form for stabilizer distillation protocols with Bell-diagonal input states (Theorem \ref{thm:standard-form}) and the properties of the canonical encoding (Definition \ref{def:canonical_enc}), a distillation protocol is proposed that maximizes the increase in fidelity for each iteration.\\
\ \\
\textbf{Protocol}: Fidelity Increase Maximizing Distillation Protocol (\textbf{FIMAX})\\
Let $\rho_{in}$ be a two-copy Bell-diagonal state ($N=2$) for prime dimension $d$.
\begin{enumerate}
    \item For each stabilizer $S$ with generator $W(g)$, $g \in \Zd^2\times \Zd^2$:
    \begin{enumerate}[(i)]
        \item Partition the error elements $e$ according to their symplectic product $s=\langle g,e\rangle$ and calculate $\Prob(\E(s))$. 
        \item For each $s$ and for each error coset $C \in \C(s)$ determine $(C_{max},s_{max}) := \arg\max \frac{\Prob(C)}{\Prob(\E(s))}$.
    \end{enumerate}
    \item Choose the stabilizer $S_{max}$, maximizing $\frac{\Prob(C_{max})}{\Prob(\E(s_{max}))}$ among all stabilizers.
    \item Alice and Bob perform stabilizer measurements for $S_{max}$ with measurement outcomes $a, b$.
    \item Bob sends $b$ to Alice. Alice declares failure of the protocol if $s_{max} \neq a-b $.
    \item Alice and Bob apply the inverse of the canonical encoding $U_c$ and $U_c^\star$ for $S_{max}$ and $S^\star_{max}$, respectively.
    \item Alice and Bob discard the first qudit and Alice applies the unique $W^\dagger_{k_{max},l_{max}} \propto T^\dagger_{C_{max}}$ to the remaining qudit.
\end{enumerate}    
\noindent One successful iteration of the protocol requires \emph{one-way} classical communication, but for further iterations, \emph{two-way} communication is required for both parties to independently determine $S_{max}$ and $C_{max}$ from the input state. Also note that the protocol is applicable to non-BDS states in two ways. First, by using the diagonal elements of the density matrix in the Bell basis as probabilities, $\prob(e) := \rho(e,e)$, to choose $S_{max}$ and $C_{max}$ to apply the remaining protocol. Second, any non-BDS can be transformed to a BDS by twirling, e.g., by a ``Weyl twirl'' \cite{popp_special_2024}, leaving all diagonal elements invariant, or a depolarizing unitary twirl \cite{horodecki_reduction_1999}. By performing such a twirl, the remaining protocol can be applied as described above. In those cases, it is not guaranteed that 
FIMAX always achieves the maximal increase in the fidelity for each iteration.\\
It remains to prove the eponymous property (cf. Theorem \ref{thm:fimax_property}) of FIMAX. The relation of all stabilizer encodings established in Section \ref{sec:errors_on_codewords} is used in Lemma \ref{thm:coset_contribution_sum} to show that the coset fidelities are probabilities. Then, the canonical encoding is shown to imply an optimal distribution of these probabilities that allows to obtain a maximal increase in fidelity (Proposition \ref{thm:canonical_encoding_maximizes_fidelity}). These results are combined to show the maximal fidelity increase in Theorem \ref{thm:fimax_property}.
\begin{lemma}
    \label{thm:coset_contribution_sum}
    Let $U$ be any encoding of a stabilizer $S$ for $N=2$ and prime dimension $d$. Let $R_C$ be the coset actions for that encoding with $C \in \C(s)$ and
    $f_C(k, l) \equiv |\myTr(\frac{1}{d}  W^\dagger_{k,l} R_C)|^2 $. 
    $~\forall (k,l)$, we have $\sum_{C\in \C(s)} f_C(k,l) = 1$.
\end{lemma} 
\begin{proof}
    Denoting the coset action of the canonical encoding by $T_C$, Proposition \ref{thm:gen_stab_encs} $(ii)$ implies $R_C = Y_1^\dagger T_C Y_2$ for some unitaries $Y_1, Y_2$. 
    Using the identity \eqref{eq:max_ent_trace_identity} and the cyclic property of the trace, we can write
    \begin{flalign*}
        f_C(k, l) = |\myTr(\frac{1}{d}  Y_2 W^\dagger_{k,l} Y_1^\dagger T_C |^2 
        = \langle \Omega_{0,0}| (T_C^\dagger \otAB \id) ~\sigma ~(T_C \otAB \id) |\Omega_{0,0} \rangle
    \end{flalign*}
    with the quantum state $\sigma \equiv (Y_2 W_{k,l} Y_1^\dagger \otAB \id) |\Omega_{0,0}\rangle \langle \Omega_{0,0} | (Y_1 W^\dagger_{k,l} Y_2^\dagger  \otAB \id)$. Lemma \ref{thm:canonical_encoding_error_actions_is_basis} implies that $T_C \otAB \id |\Omega_{0,0} \rangle$ are basis states and thus $
        \sum_{C \in \C(s)} \langle \Omega_{0,0}| (T_C^\dagger \otAB \id) ~\sigma ~(T_C \otAB \id) |\Omega_{0,0} \rangle = \Tr(\sigma) = 1$.
\end{proof}
\begin{prop}
\label{thm:canonical_encoding_maximizes_fidelity}
Let $S$ be a stabilizer for $N=2$ and prime $d$. Let $U_c$ be the canonical encoding and $V$ be another arbitrary encoding. Denote the output state fidelities as $\F^{U_c}(k,l)$ and $\F^V(k,l)$, respectively. Then $\exists(k_{max}, l_{max})$ such that $\F^{U_c}(k_{max}, l_{max}) \geq \F^V(k,l)~ \forall (k,l)$.
\end{prop}
\begin{proof}
With Corollary \ref{thm:fidelity_coset_contribution} and Lemma \ref{thm:coset_contribution_sum}, we have $f_C(k, l) \geq 0$, $\sum_C f_C(k, l) = 1$ and $\Prob(C) \geq 0$ for all $C$. Therefore, the Karush-Kuhn-Tucker conditions \cite{karush_minima_1939, kuhn_nonlinear_2014} are satisfied. In consequence,
$\F^V(k,l) = \sum_{C \in \C(s)} \frac{\Prob(C)}{\Prob(\E(s))} f_C(k,l) \leq \frac{\Prob(C_{max})}{\Prob(\E(s))} ~~\forall V ~\forall (k,l)$, 
where $C_{max} = \arg\max \Prob(C)$ is the error coset with maximum probability. Consider the canonical encoding $U_c$. By Lemma \ref{thm:canonical_error_actions}, the corresponding coset actions are of the form $T_C \propto W_{t(C),-s_1(C)}$. Define $(k_{max},l_{max}) := (t(C_{max}), -s_1(C_{max}))$ and thus $T_{C_{max}} \propto W_{k_{max}, l_{max}}$. The definition of $f_C$ in Corollary \ref{thm:fidelity_coset_contribution} and Proposition \ref{thm:coset_contribution_properties} (i) imply $f_{C_{max}}(k, l) = \delta_{(k,l), (k_{max}, l_{max})}$ and thus $\F^{U_c}(k_{max}, l_{max}) = \frac{\Prob(C_{max})}{\Prob(\E(s))} \geq F^V(k,l) ~~ \forall V~ \forall (k,l)$.
\end{proof}
\begin{theorem} \label{thm:fimax_property} Let $\rho_{in}$ be a Bell-diagonal state of prime dimension $d$. Among all two-copy stabilizer-based distillation protocols, the FIMAX protocol maximizes the increase in the fidelity for a single iteration.
\end{theorem}
\begin{proof}
Using the standard form \ref{thm:standard-form}, Proposition \ref{thm:canonical_encoding_maximizes_fidelity} proves that the canonical encoding for a given stabilizer maximizes the fidelity gain for a single iteration of a stabilizer protocol among all encodings and that the achievable fidelity is precisely $\frac{\Prob(C)}{\Prob(\E(s))}$. Therefore, using the stabilizer that maximizes this quantity with the canonical encoding for specific measurement outcomes implies the maximal fidelity between the output state and some Bell state $|\Omega_{k,l}\rangle$. Applying the inverse of the corresponding Weyl operator $W^{\dagger}_{k_{max}, l_{max}}$ therefore maximizes the fidelity between the output state and $|\Omega_{0,0}\rangle$.
\end{proof}

\noindent We close this section with a brief comment on the computational complexity on the application of FIMAX. Given the probability distribution on all error elements for $N=2$ via the Bell-diagonal input state, the protocol requires determining all nontrivial stabilizers $S \neq \{ \mathbbm{1} \},S\neq\E_2$, which is equivalent to finding all subgroups of $\Zd \times \Zd$. Since $d$ is prime, all subgroups are cyclic, each $e \in \Zd \times \Zd, e \neq (\Vec{0}, \Vec{0})$ generates a subgroup of $d$ elements and each element is part of only one subgroup. Each subgroup contains $d-1$ elements in addition to the neutral element $(\Vec{0}, \Vec{0})$. Since the subgroups are disjoint up to the neutral element, the combined number of elements contained in $n_S$ subgroups is $n_S(d-1)+1$. Equating this to the total number of $d^4$ elements in $\E_2$, we conclude that there are $n_S = (d^2+1)(d+1)$ stabilizers. For each stabilizer, there are $d^3$ distinct cosets. Consequently, to find the stabilizer and coset maximizing $\frac{\Prob(C)}{\Prob(\E(s))}$ in the second step of FIMAX, $d^3(d^2+1)(d+1)$ probabilities have to be calculated in each iteration. This implies that the complexity for FIMAX is of polynomial order in $d$. Note that this only holds for a fixed number of copies $N=2$. For more general stabilizer distillation protocols, the number of stabilizers to consider generally grows exponentially with the number of used copies $N$.

\subsection{Efficacy of the FIMAX protocol and comparison to other protocols}
\label{sec:numeric_comparison}
In this section, the efficacy of the FIMAX protocol is demonstrated. The proposed protocol is compared to well-established two-copy recurrence protocols regarding the minimal required fidelity and the protocol efficiency. We compare it to the generalization of the BBPSSW protocol  \cite{bennett_purification_1996, horodecki_reduction_1999}, the DEJMPS protocol \cite{deutsch_quantum_1996} ($d=2$ only) and the ADGJ protocol \cite{alber_efficient_2001} to $d$-level systems. In addition, we compare to the so-called ``P1-or-P2'' protocol (here named ``P12'' protocol) \cite{miguel-ramiro_efficient_2018}, known to outperform the BBPSSW and ADGJ protocol in efficiency and minimal required fidelity for $d=3$. Note, that these analyses are intended to show the potential of stabilizer-based distillation without providing a complete evaluation regarding its performance in general settings. Additional numerical analyses regarding the performance of FIMAX can be found in \cite{popp_low-fidelity_2025}. All applied methods are implemented as open source software \cite{popp_belldiagonalqudits_2023}. \\
First, we compare the distillation efficiency in dimensions $d=2$ and $d=3$. Given a target fidelity, the efficiency is defined as the inverse of the expected number of input states required to produce one output pair with fidelity larger than the target fidelity. Using two copies for each iteration with success probability $p_i$ and requiring $N_{it}$ iterations to reach the target fidelity, the efficiency is $2^{-N_{it}}\prod_i p_i $. Here, we choose a target fidelity of $0.999$. If the target fidelity cannot be reached, the efficiency of the protocol is zero. 
In Figure  \ref{fig:iso_eff} we analyze isotropic states. This family is defined as mixtures of the target state with the maximally mixed state $\pi_{mm}$, i.e., $\rho_{iso}(p) := p~|\Omega_{0,0}\rangle\langle \Omega_{0,0}| + (1-p)~\pi_{mm}$. The proposed FIMAX protocol can distill all states with fidelity $>1/d$ and is more efficient in wide ranges of initial fidelity than the other protocols (except for the DEJMPS protocol in $d=2$, which has the same efficiency). Especially in the low fidelity regime, the efficiency of FIMAX protocol can be more than a magnitude higher than for each of the other protocols in $d=3$. No fidelities are observed for which the efficiency of FIMAX is lower than any of the other protocols. 
\begin{figure}[H]
    \vspace{-2em}
    \centering
    \subfloat[Isotropic state family, $d=2$]{\includegraphics[width=0.495\linewidth]{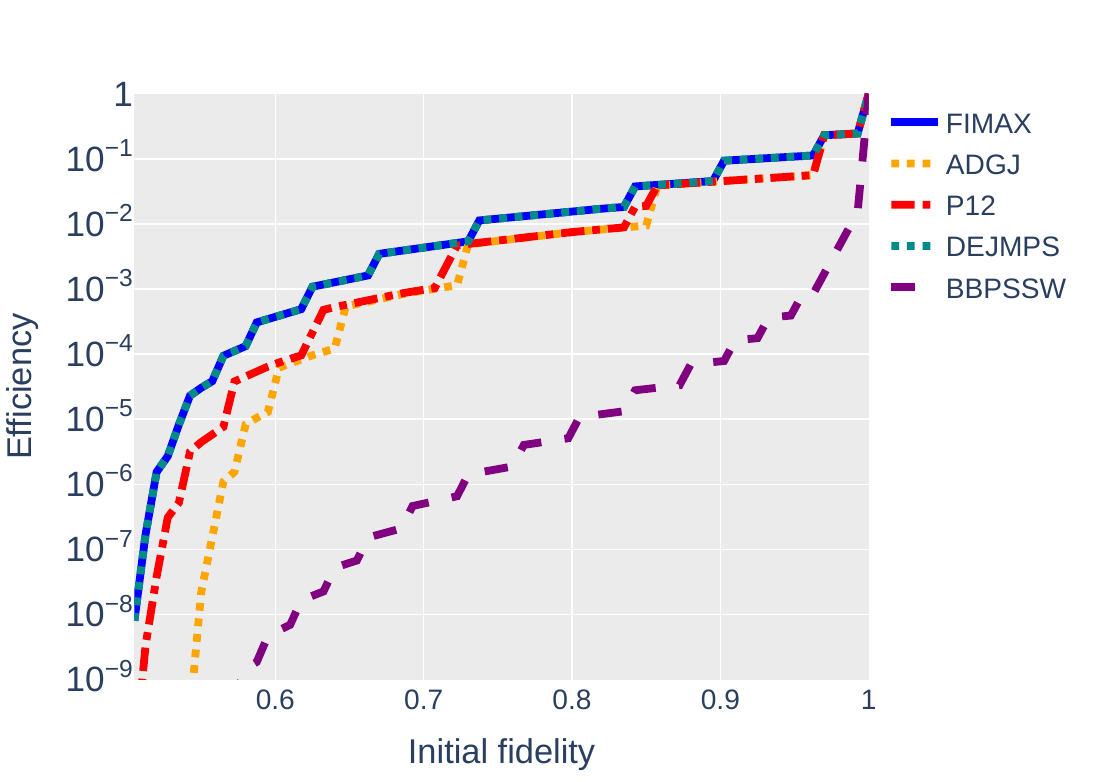}} 
    \label{fig:iso_2_eff}
    \subfloat[Isotropic state family, $d=3$]{\includegraphics[width=0.495\linewidth]{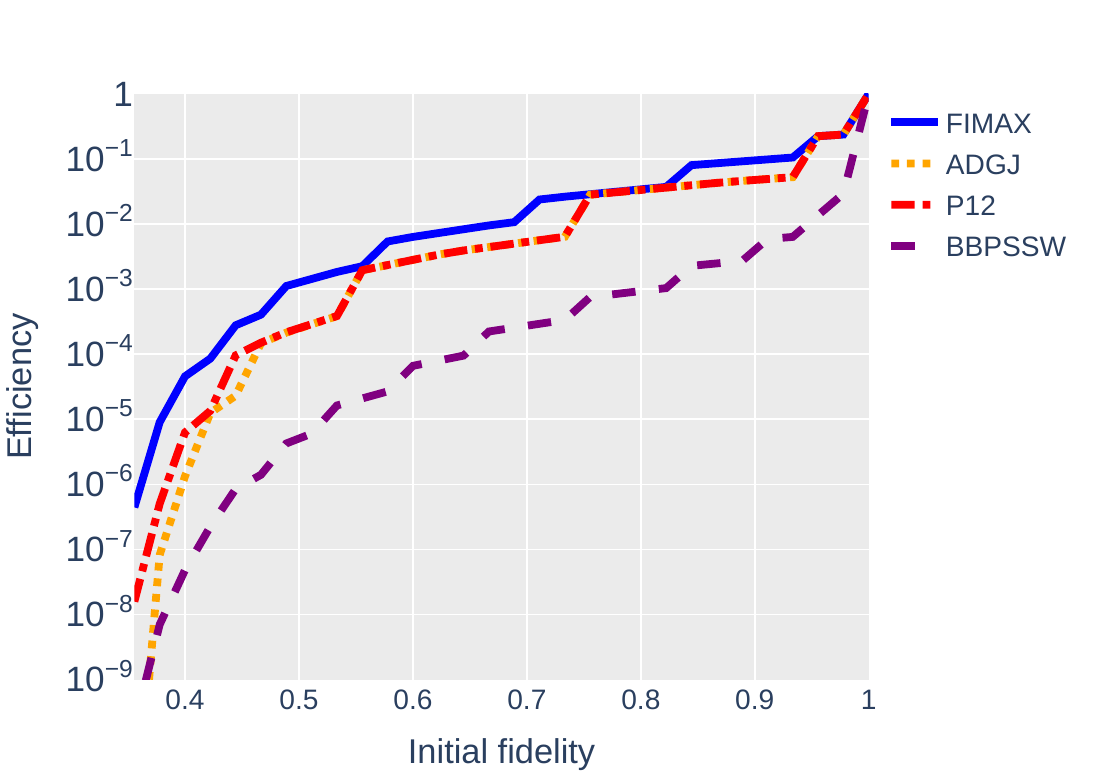}}
    \label{fig:iso_3_eff}
    \caption{
    Protocol comparison of distillation efficiencies depending on the fidelity of isotropic input states.
    }
    \label{fig:iso_eff}
\end{figure} \noindent
The mean efficiency of randomly generated pure states that are grouped by their fidelity is depicted in Figure \ref{fig:random_eff}. We show the efficiency in the fidelity range that allows sampling random pure states with appropriate computational effort. Note, that it is very unlikely to sample uniformly distributed states with fidelity higher than a certain value, depending on the dimension. We choose the fidelity ranges to be $[0,0.9]$ for $d=2$ and $[0,0.6]$ for $d=3$. Within these ranges, we sample bins of $1000$ states, where each bin corresponds to a unique fidelity value, rounded to two digits, and calculate the mean efficiency. To estimate the numerical error given the limited number of samples, we calculate the standard deviation for the efficiency $\sigma_{bin}$ in each bin. The error for the mean can then be estimated as $\sigma_{bin}/\sqrt{1000}$. Relative to the mean, the maximum error for all protocols and all fidelity bins is $<1\%$ for $d=2$ and $<25\%$ for $d=3$. Notably, this error is significantly higher for $d=3$. FIMAX is applied with prior twirl to obtain Bell-diagonal states. Interestingly, FIMAX again performs best despite the additional twirling operation. Fidelity regions below $1/d$ are visible, in which the proposed protocol is the only one capable of distillation for the limited set of $1000$ analyzed states per bin.
\begin{figure}[H]
    \vspace{-2em}
    \centering
    \subfloat[Efficiency of random pure bipartite states, $d=2$]{\includegraphics[width=0.495\linewidth]{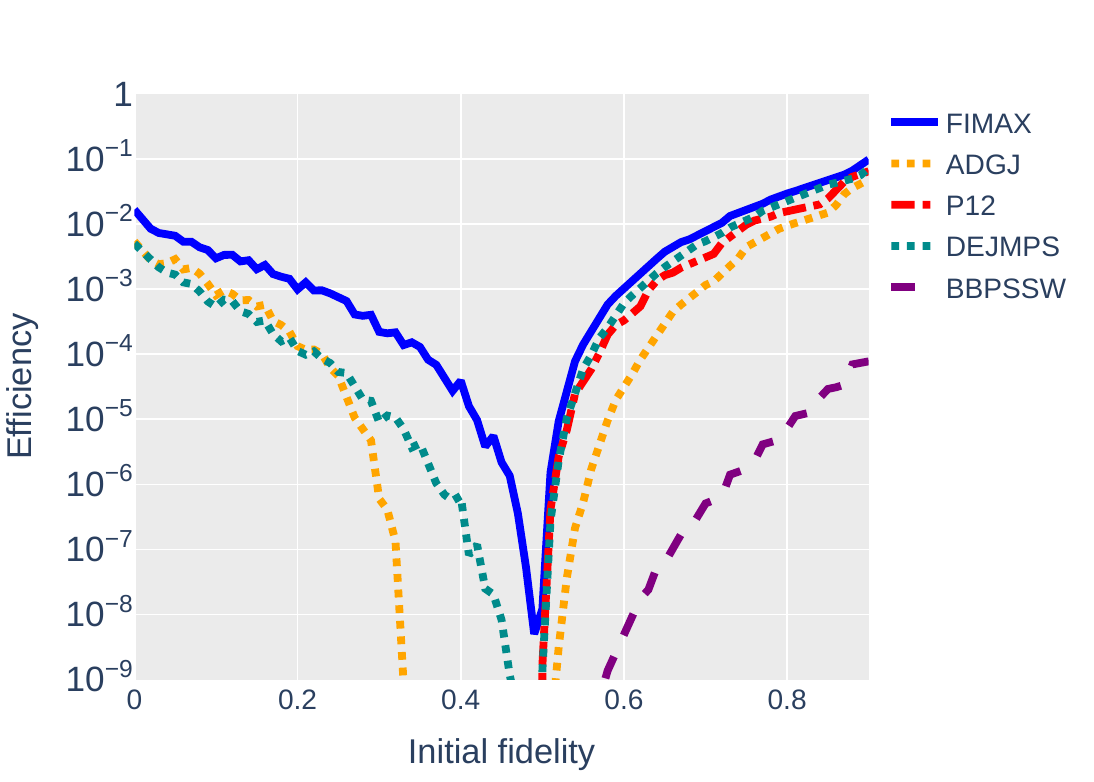}} 
    \label{fig:random_2_eff}
    \subfloat[Efficiency of random pure bipartite states, $d=3$]{\includegraphics[width=0.495\linewidth]{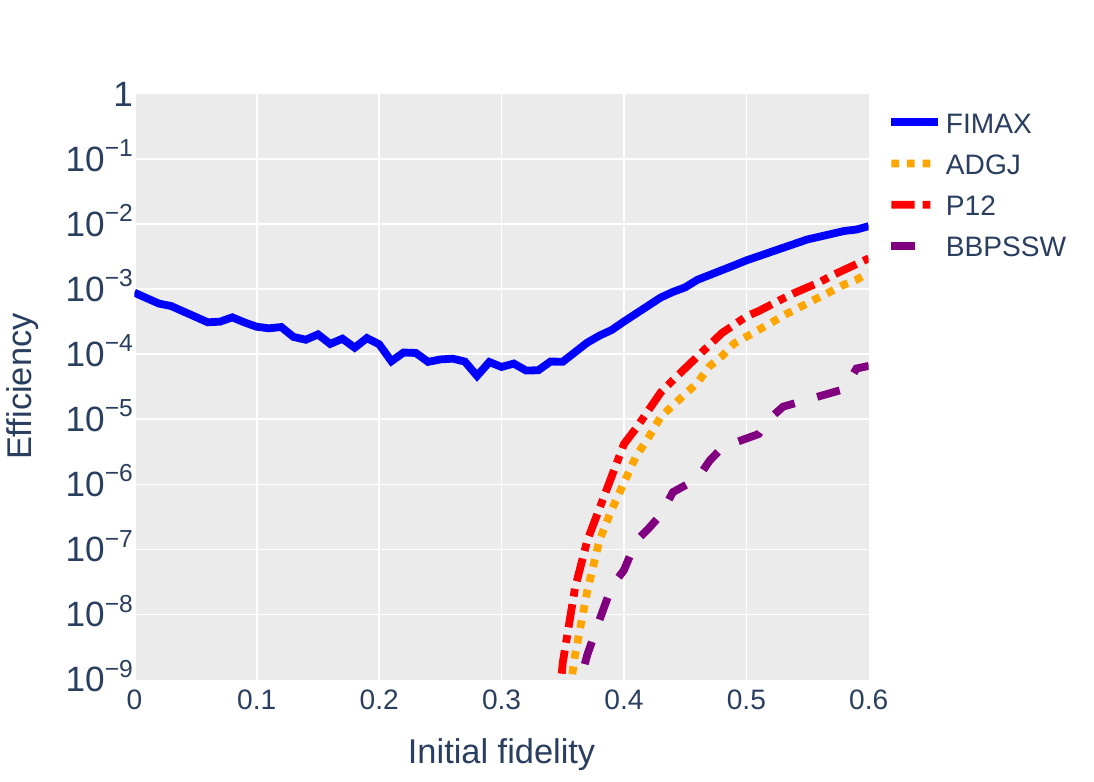}}
    \label{fig:random_3_eff}
    \caption{
    Efficiency comparison depending on the fidelity. States are grouped in bins of $1000$ states by their fidelity, rounded to two digits, and the mean efficiency for all protocols is determined.
    }
    \label{fig:random_eff}
\end{figure}
\noindent
In Figure  \ref{fig:iso_fid}, the effect of the protocols for specific low-fidelity isotropic states in $d=2$ and $d=3$ is visualized. For $d=2$, FIMAX and DEJMPS have equal fidelity increase for this state. ADGJ fails to distill, while BBPSSW generally increases the fidelity, but less than FIMAX. P12 has iterations, which do not significantly increase the fidelity. FIMAX reaches the target fidelity with the least number of iterations.
\begin{figure}[H]
    \vspace{-2em}
    \centering
    \subfloat[Isotropic state, $p=0.35, d=2$]{\includegraphics[width=0.495\linewidth]{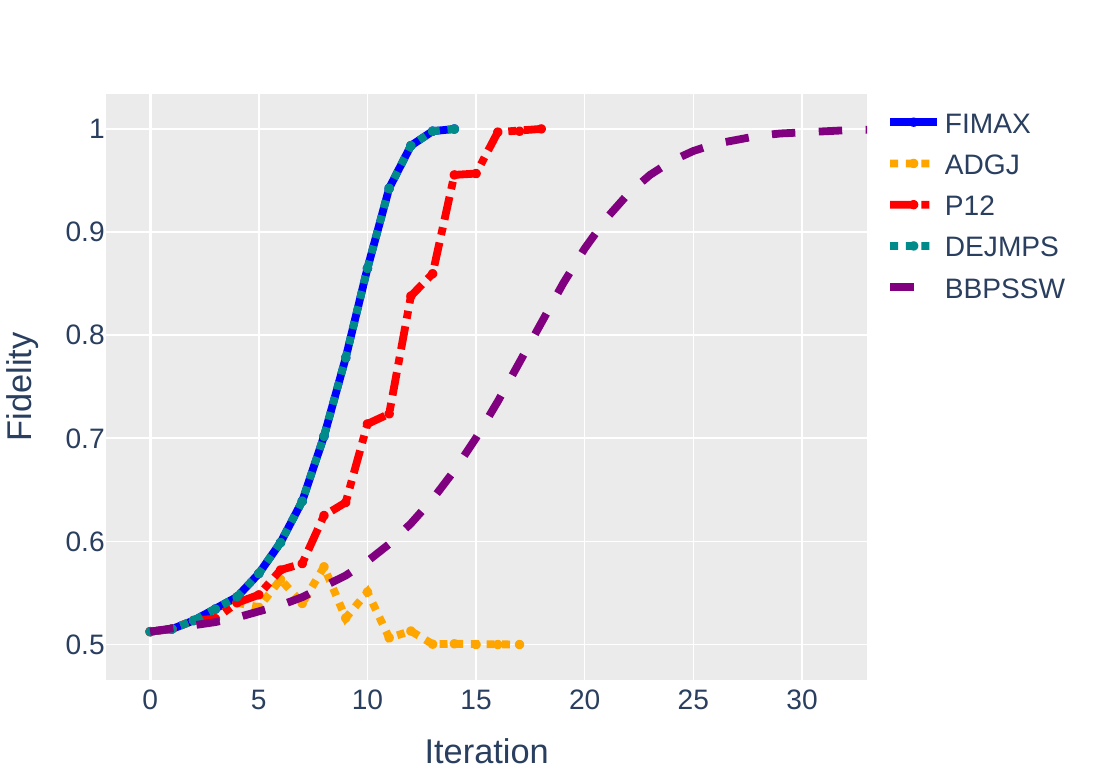}} 
    \label{fig:iso_2_fid}
    \subfloat[Isotropic state, $p=0.26, d=3$]{\includegraphics[width=0.495\linewidth]{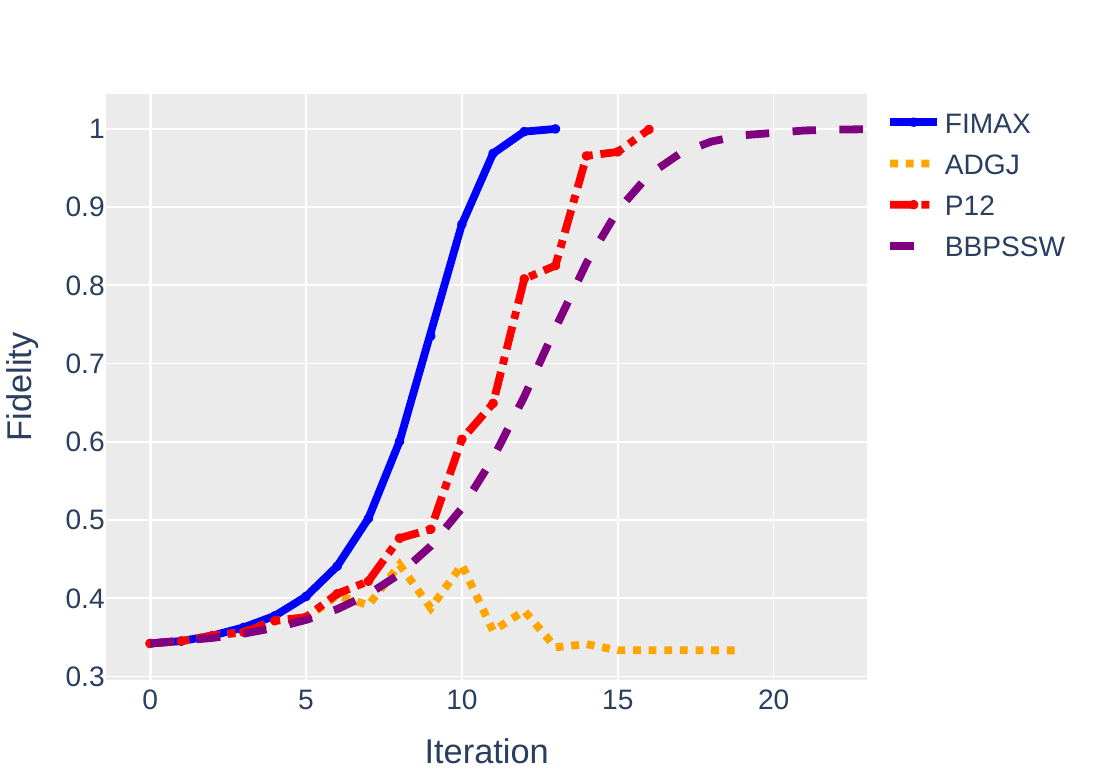}}
    \label{fig:iso_3_fid}
    \caption{Protocol comparison of the iterative fidelity increase for a low fidelity isotropic input state.}
    \label{fig:iso_fid}
\end{figure}
\noindent
Finally, Figure  \ref{fig:ol} demonstrates that FIMAX can distill states with fidelity $<1/d$ with high efficiency. For $d=3$, we define the states $\rho_{ol}(p) := p~\sigma + (1-p)~\pi_{mm} \text{ with }\sigma := 1/3(|\Omega_{0,0}\rangle\langle \Omega_{0,0}| + |\Omega_{1,0}\rangle\langle \Omega_{1,0}| + |\Omega_{0,1}\rangle\langle \Omega_{0,1}|)$, the so-called ``off-line states''. All of these states have fidelity $\leq 1/d$ and none of the other protocols can distill any state of this family. Figure \ref{fig:ol}(a) demonstrates that the FIMAX protocol can distill all off-line states with initial fidelity $>0.25$. Figure  \ref{fig:ol}(b) shows the iterative increase in fidelity together with the probability of success for each iteration. Interestingly, the protocol increases the fidelity in the first iteration to a value $>1/d$.
\begin{figure}[tph]
    \vspace{-2em}
    \centering
    \subfloat[Off-line state family, $d=3$]{\includegraphics[width=0.495\linewidth]{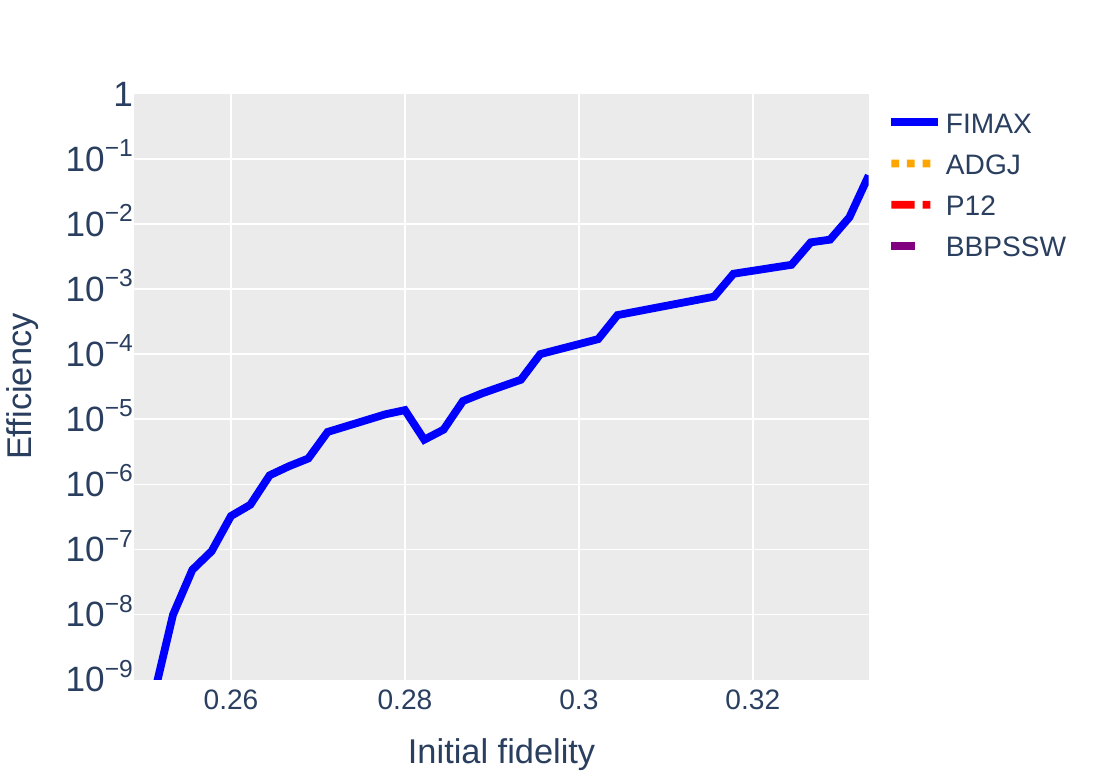}} 
    \label{fig:ol_3_eff}
    \subfloat[Off-line state, $p=0.7, d=3$]{\includegraphics[width=0.495\linewidth]{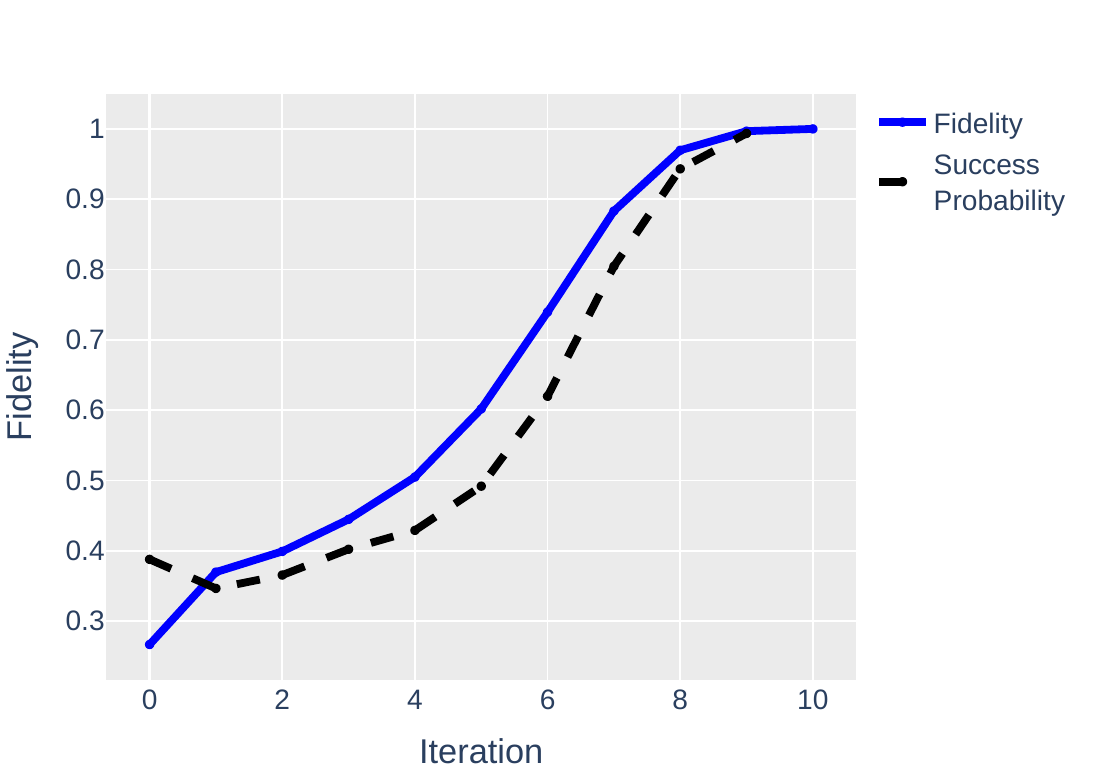}}
    \label{fig:ol_3_fid}
    \caption{Distillation efficiency for off-line states in dependence on the fidelity (a). Fidelities and success probabilities for each FIMAX iteration for a low-fidelity off-line input state (b).}
    \label{fig:ol}
\end{figure}

\newpage
\section{Discussion and Conclusion}
\label{sec:discussion_and_conclusion}
In this work, we analyzed the action of the stabilizer-based distillation procedure in prime dimension to derive a standard form of the output state, making the effect of the adjustable parameters of the protocol transparent. We leveraged this standard form to propose FIMAX, a fidelity increase maximizing distillation protocol, that demonstrates superior efficacy compared to other well-established protocols regarding efficiency and minimal fidelity requirements.\\
It was shown how the effective action of Weyl-Heisenberg errors depends on the chosen codewords and how all encodings are analytically related. The group properties of the  Weyl-Heisenberg errors were extended to the so-called error action operators, representing the effective action of a Weyl-Heisenberg error in the stabilizer protocol. 
It was further demonstrated how the stabilizer implies a decomposition of the set of errors given by its cosets and the outcomes of stabilizer measurements. 
Combining those general insights, we derived a standard form of the output state, rendering the role of input state, stabilizer, encoding, and measurement evident and making the calculation of all protocol output fidelities possible.  \\
Focusing on two-copy stabilizer distillation in prime dimension, we introduced a canonical encoding, for which we find that the effective actions of errors are again of Weyl-Heisenberg form. Leveraging the standard form and the properties of the canonical encoding, we proposed the distillation protocol FIMAX and proved that among all two-copy stabilizer protocols in prime dimension, it implies the maximal increase in fidelity for Bell-diagonal states in each iteration.\\
Finally, we compared the new protocol to prominent recurrence protocols, namely BBPSSW \cite{bennett_purification_1996, horodecki_reduction_1999}, DEJMPS \cite{deutsch_quantum_1996}, ADGJ \cite{alber_efficient_2001} and ``P1-or-P2/P12'' \cite{miguel-ramiro_efficient_2018}, that have been shown to have good efficiency and also allow for the distillation of low-fidelity states. FIMAX demonstrated the best results regarding distillation efficiency and disability of both Bell-diagonal and, curiously, also non-Bell-diagonal states in all numerical investigations. Due to the limited number of samples and state families, these results do not prove general superiority, requiring more analyses with higher sample sizes, especially for $d\geq3$. However, the reported results clearly indicate the potential of the developed formalism and the FIMAX protocol. Further results confirming this potential with focus on entanglement distillation of low-fidelity states can be found in \cite{popp_low-fidelity_2025}.\\
The developed theory of stabilizer-based entanglement distillation aims to enable future research in the construction of new distillation protocols and in the general problem of distillability of mixed states. The successful application to the two-copy case in prime dimension illustrates how the derived standard form helps to develop stabilizer protocols and analyze their properties. Many existing protocols, including BBPSSW, DEJMPS and P12, are equivalent or strongly related to a specific stabilizer protocol (see, e.g., \cite{matsumoto_conversion_2003}). Interestingly, another generalization of such recurrence-type protocols has been suggested, so-called permutation-based schemes \cite{dehaene_local_2003,bombin_entanglement_2005}. Both approaches are related by their symplectic structure manifesting in the investigated  properties of stabilizers and their encodings on the one side, and in the form of permutation matrices that correspond to local operations on the other side. The standard stabilizer and permutation protocols have been shown to be equivalent regarding the output fidelity in the case of Bell-diagonal input states for $d=2$ \cite{hostens_equivalence_2004}. For general dimension $d$, however, this equivalence is not expected. All permutation-based protocols map the set of Bell-diagonal states onto itself. However, with the presented results, one easily finds stabilizer codes that imply a mapping to non-Bell diagonal states. Conversely, it is unclear whether every permutation protocol can be realized by a stabilizer protocol with suitable encoding. This would imply that the class of stabilizer-distillation schemes is strictly larger than the permutation-based one. Further research in this direction could contribute to the construction of optimal protocols for specific state families. \\
Interestingly, the proposed FIMAX protocol does map the set of Bell-diagonal states onto itself, implying that the canonical encoding is part of the Clifford group and therefore can be efficiently constructed with quantum gates for $d=2$ \cite{watanabe_improvement_2006}. Whether this also holds for $d \geq3$ remains an open question for future research. Extending the developed methods to the non-prime dimensional regime also provides an interesting challenge for the future. While generalization to prime-power dimensions may be possible by following the theory of nonbinary quantum stabilizer codes \cite{ashikhmin_nonbinary_2001}, other dimensions may be challenging due to the more complicated spectral properties of corresponding operators and their group structure. \\
The numerical results regarding the performance of FIMAX clearly demonstrate that the developed stabilizer approach offers great potential for effective distillation. While further investigations are needed for a general performance evaluation of  FIMAX, the results clearly indicate high efficiency compared to the other protocols for certain state families. Interestingly, FIMAX also exhibits strong performance for pure states if a twirl to Bell-diagonal form is prepended. This is surprising, as such an operation generally reduces the fidelity due to its data processing inequality \cite{khatri_principles_2024}. Investigating the impact of twirling on distillability and conducting a comparative analysis of performance relative to a protocol executed without prior twirling constitutes an intriguing subject for future research. Notable performance is especially evident in the distillation of low-fidelity Bell-diagonal states, as also recently affirmed \cite{popp_low-fidelity_2025}.
These findings indicate the potential utility of the presented approach in addressing the broader challenges of general distillability and the phenomenon of bound entanglement \cite{hiesmayr_bipartite_2025}. Related research directions include further development of the theory, protocols and numerical investigations to the multi-copy, non-Bell-diagonal and non-prime dimensional regimes.
\newpage
\printbibliography

@incollection{kuhn_nonlinear_2014,
	address = {Basel},
	title = {Nonlinear {Programming}},
	isbn = {978-3-0348-0439-4},
	url = {https://doi.org/10.1007/978-3-0348-0439-4_11},
	language = {en},
	urldate = {2025-12-10},
	booktitle = {Traces and {Emergence} of {Nonlinear} {Programming}},
	publisher = {Springer},
	author = {Kuhn, Harold W. and Tucker, Albert W.},
	editor = {Giorgi, Giorgio and Kjeldsen, Tinne Hoff},
	year = {2014},
	doi = {10.1007/978-3-0348-0439-4_11},
	pages = {247--258},
}

@misc{matsumoto_breeding_2024,
	title = {Breeding protocols are advantageous for finite-length entanglement distillation},
	url = {http://arxiv.org/abs/2401.02265},
	doi = {10.48550/arXiv.2401.02265},
	urldate = {2025-12-10},
	publisher = {arXiv},
	author = {Matsumoto, Ryutaroh},
	month = feb,
	year = {2024},
	note = {arXiv:2401.02265 [quant-ph]},
}

@misc{wilde_quantum_2008,
	title = {Quantum {Coding} with {Entanglement}},
	url = {http://arxiv.org/abs/0806.4214},
	doi = {https://doi.org/10.48550/arXiv.0806.4214},
	language = {en},
	urldate = {2024-06-25},
	publisher = {arXiv},
	author = {Wilde, Mark M.},
	month = jun,
	year = {2008},
	note = {arXiv:0806.4214 [quant-ph]},
}

@article{gottesman_stabilizer_1997,
	title = {Stabilizer {Codes} and {Quantum} {Error} {Correction}},
	url = {http://arxiv.org/abs/quant-ph/9705052},
	doi = {https://doi.org/10.48550/arXiv.quant-ph/9705052},
	language = {en},
	urldate = {2024-02-06},
	author = {Gottesman, Daniel},
	month = may,
	year = {1997},
	note = {arXiv:quant-ph/9705052},
}

@article{alber_efficient_2001,
	title = {Efficient bipartite quantum state purification in arbitrary dimensional {Hilbert} spaces},
	volume = {34},
	issn = {0305-4470},
	url = {https://doi.org/10.1088/0305-4470/34/42/307},
	doi = {10.1088/0305-4470/34/42/307},
	language = {en},
	number = {42},
	urldate = {2025-12-10},
	journal = {Journal of Physics A: Mathematical and General},
	author = {Alber, Gernot and Delgado, Aldo and Gisin, Nicolas and Jex, Igor},
	month = oct,
	year = {2001},
	pages = {8821},
}

@article{popp_low-fidelity_2025,
	title = {Low-fidelity entanglement distillation with {FIMAX}},
	volume = {23},
	issn = {0219-7499},
	url = {https://www.worldscientific.com/doi/10.1142/S0219749925500170},
	doi = {10.1142/S0219749925500170},
	number = {06},
	urldate = {2025-10-24},
	journal = {International Journal of Quantum Information},
	author = {Popp, Christopher and Sutter, Tobias C. and Hiesmayr, Beatrix C.},
	month = sep,
	year = {2025},
	note = {Publisher: World Scientific Publishing Co.},
	pages = {2550017},
}

@article{hiesmayr_bipartite_2025,
	title = {Bipartite bound entanglement},
	volume = {23},
	issn = {0219-7499},
	url = {https://www.worldscientific.com/doi/10.1142/S0219749925300037},
	doi = {10.1142/S0219749925300037},
	number = {05},
	urldate = {2025-10-23},
	journal = {International Journal of Quantum Information},
	author = {Hiesmayr, Beatrix C. and Popp, Christopher and Sutter, Tobias C.},
	month = aug,
	year = {2025},
	note = {Publisher: World Scientific Publishing Co.},
	pages = {2530003},
}

@article{ashikhmin_nonbinary_2001,
	title = {Nonbinary quantum stabilizer codes},
	volume = {47},
	issn = {1557-9654},
	url = {https://ieeexplore.ieee.org/document/959288},
	doi = {10.1109/18.959288},
	number = {7},
	urldate = {2025-10-22},
	journal = {IEEE Transactions on Information Theory},
	author = {Ashikhmin, A. and Knill, E.},
	month = nov,
	year = {2001},
	pages = {3065--3072},
}

@article{popp_belldiagonalqudits_2023,
	title = {{BellDiagonalQudits}: {A} package for entanglement analyses of mixed maximally entangled qudits},
	volume = {8},
	issn = {2475-9066},
	shorttitle = {{BellDiagonalQudits}},
	url = {https://joss.theoj.org/papers/10.21105/joss.04924},
	doi = {10.21105/joss.04924},
	language = {en},
	number = {81},
	urldate = {2024-02-06},
	journal = {Journal of Open Source Software},
	author = {Popp, Christopher},
	month = jan,
	year = {2023},
	pages = {4924},
}

@article{matsumoto_conversion_2003,
	title = {Conversion of a general quantum stabilizer code to an entanglement distillation protocol},
	volume = {36},
	issn = {0305-4470, 1361-6447},
	url = {http://arxiv.org/abs/quant-ph/0209091},
	doi = {10.1088/0305-4470/36/29/316},
	language = {en},
	number = {29},
	urldate = {2024-02-06},
	journal = {Journal of Physics A: Mathematical and General},
	author = {Matsumoto, Ryutaroh},
	month = jul,
	year = {2003},
	note = {arXiv:quant-ph/0209091},
	pages = {8113--8127},
}

@misc{khatri_principles_2024,
	title = {Principles of {Quantum} {Communication} {Theory}: {A} {Modern} {Approach}},
	shorttitle = {Principles of {Quantum} {Communication} {Theory}},
	url = {http://arxiv.org/abs/2011.04672},
	doi = {10.48550/arXiv.2011.04672},
	urldate = {2024-08-19},
	publisher = {arXiv},
	author = {Khatri, Sumeet and Wilde, Mark M.},
	month = feb,
	year = {2024},
	note = {arXiv:2011.04672 [cond-mat, physics:hep-th, physics:math-ph, physics:quant-ph]},
}

@article{zhou_purification_2020,
	title = {Purification of the residual entanglement},
	volume = {28},
	copyright = {© 2020 Optical Society of America},
	issn = {1094-4087},
	url = {https://opg.optica.org/oe/abstract.cfm?uri=oe-28-2-2291},
	doi = {10.1364/OE.383499},
	language = {EN},
	number = {2},
	urldate = {2024-08-07},
	journal = {Optics Express},
	author = {Zhou, Lan and Zhong, Wei and Sheng, Yu-Bo},
	month = jan,
	year = {2020},
	note = {Publisher: Optica Publishing Group},
	pages = {2291--2301},
}

@article{yan_measurement-based_2022,
	title = {Measurement-based logical qubit entanglement purification},
	volume = {105},
	url = {https://link.aps.org/doi/10.1103/PhysRevA.105.062418},
	doi = {10.1103/PhysRevA.105.062418},
	number = {6},
	urldate = {2024-08-07},
	journal = {Physical Review A},
	author = {Yan, Pei-Shun and Zhou, Lan and Zhong, Wei and Sheng, Yu-Bo},
	month = jun,
	year = {2022},
	note = {Publisher: American Physical Society},
	pages = {062418},
}

@article{rains_nonbinary_1999,
	title = {Nonbinary quantum codes},
	volume = {45},
	issn = {1557-9654},
	url = {https://ieeexplore.ieee.org/document/782103},
	doi = {10.1109/18.782103},
	number = {6},
	urldate = {2024-08-01},
	journal = {IEEE Transactions on Information Theory},
	author = {Rains, E.M.},
	month = sep,
	year = {1999},
	note = {Conference Name: IEEE Transactions on Information Theory},
	pages = {1827--1832},
}

@phdthesis{karush_minima_1939,
	title = {Minima of functions of several variables with inequalities as side conditions},
	url = {https://catalog.lib.uchicago.edu/vufind/Record/4111654},
	urldate = {2024-07-31},
	author = {Karush, William},
	year = {1939},
	note = {OCLC: 43268508},
}

@misc{hostens_equivalence_2004,
	title = {The equivalence of two approaches to the design of entanglement distillation protocols},
	url = {http://arxiv.org/abs/quant-ph/0406017},
	doi = {10.48550/arXiv.quant-ph/0406017},
	urldate = {2024-07-22},
	publisher = {arXiv},
	author = {Hostens, Erik and Dehaene, Jeroen and De Moor, Bart},
	month = jun,
	year = {2004},
	note = {arXiv:quant-ph/0406017},
}

@article{dur_entanglement_2007,
	title = {Entanglement purification and quantum error correction},
	volume = {70},
	issn = {0034-4885},
	url = {https://dx.doi.org/10.1088/0034-4885/70/8/R03},
	doi = {10.1088/0034-4885/70/8/R03},
	language = {en},
	number = {8},
	urldate = {2024-07-22},
	journal = {Reports on Progress in Physics},
	author = {Dür, W. and Briegel, H. J.},
	month = jul,
	year = {2007},
	pages = {1381},
}

@article{horodecki_reduction_1999,
	title = {Reduction criterion of separability and limits for a class of distillation protocols},
	volume = {59},
	url = {https://link.aps.org/doi/10.1103/PhysRevA.59.4206},
	doi = {10.1103/PhysRevA.59.4206},
	number = {6},
	urldate = {2024-07-22},
	journal = {Physical Review A},
	author = {Horodecki, Michał and Horodecki, Paweł},
	month = jun,
	year = {1999},
	note = {Publisher: American Physical Society},
	pages = {4206--4216},
}

@article{werner_all_2001,
	title = {All teleportation and dense coding schemes},
	volume = {34},
	issn = {0305-4470},
	url = {https://dx.doi.org/10.1088/0305-4470/34/35/332},
	doi = {10.1088/0305-4470/34/35/332},
	language = {en},
	number = {35},
	urldate = {2024-07-22},
	journal = {Journal of Physics A: Mathematical and General},
	author = {Werner, R. F.},
	month = aug,
	year = {2001},
	pages = {7081},
}

@article{briegel_measurement-based_2009,
	title = {Measurement-based quantum computation},
	volume = {5},
	copyright = {2009 Springer Nature Limited},
	issn = {1745-2481},
	url = {https://www.nature.com/articles/nphys1157},
	doi = {10.1038/nphys1157},
	language = {en},
	number = {1},
	urldate = {2024-07-18},
	journal = {Nature Physics},
	author = {Briegel, H. J. and Browne, D. E. and Dür, W. and Raussendorf, R. and Van den Nest, M.},
	month = jan,
	year = {2009},
	note = {Publisher: Nature Publishing Group},
	pages = {19--26},
}

@article{horodecki_mixed-state_1998,
	title = {Mixed-{State} {Entanglement} and {Distillation}: {Is} there a “{Bound}” {Entanglement} in {Nature}?},
	volume = {80},
	issn = {0031-9007, 1079-7114},
	shorttitle = {Mixed-{State} {Entanglement} and {Distillation}},
	url = {https://link.aps.org/doi/10.1103/PhysRevLett.80.5239},
	doi = {10.1103/PhysRevLett.80.5239},
	language = {en},
	number = {24},
	urldate = {2024-02-05},
	journal = {Physical Review Letters},
	author = {Horodecki, Michał and Horodecki, Paweł and Horodecki, Ryszard},
	month = jun,
	year = {1998},
	pages = {5239--5242},
}

@article{bennett_concentrating_1996,
	title = {Concentrating partial entanglement by local operations},
	volume = {53},
	issn = {1050-2947, 1094-1622},
	url = {https://link.aps.org/doi/10.1103/PhysRevA.53.2046},
	doi = {10.1103/PhysRevA.53.2046},
	language = {en},
	number = {4},
	urldate = {2024-02-05},
	journal = {Physical Review A},
	author = {Bennett, Charles H. and Bernstein, Herbert J. and Popescu, Sandu and Schumacher, Benjamin},
	month = apr,
	year = {1996},
	pages = {2046--2052},
}

@article{deutsch_quantum_1996,
	title = {Quantum {Privacy} {Amplification} and the {Security} of {Quantum} {Cryptography} over {Noisy} {Channels}},
	volume = {77},
	url = {https://link.aps.org/doi/10.1103/PhysRevLett.77.2818},
	doi = {10.1103/PhysRevLett.77.2818},
	number = {13},
	urldate = {2024-07-17},
	journal = {Physical Review Letters},
	author = {Deutsch, David and Ekert, Artur and Jozsa, Richard and Macchiavello, Chiara and Popescu, Sandu and Sanpera, Anna},
	month = sep,
	year = {1996},
	note = {Publisher: American Physical Society},
	pages = {2818--2821},
}

@article{watanabe_improvement_2006,
	title = {Improvement of stabilizer-based entanglement distillation protocols by encoding operators},
	volume = {39},
	issn = {0305-4470},
	url = {https://dx.doi.org/10.1088/0305-4470/39/16/013},
	doi = {10.1088/0305-4470/39/16/013},
	language = {en},
	number = {16},
	urldate = {2024-06-24},
	journal = {Journal of Physics A: Mathematical and General},
	author = {Watanabe, Shun and Matsumoto, Ryutaroh and Uyematsu, Tomohiko},
	month = mar,
	year = {2006},
	pages = {4273},
}

@article{bombin_entanglement_2005,
	title = {Entanglement distillation protocols and number theory},
	volume = {72},
	issn = {1050-2947, 1094-1622},
	url = {https://link.aps.org/doi/10.1103/PhysRevA.72.032313},
	doi = {10.1103/PhysRevA.72.032313},
	language = {en},
	number = {3},
	urldate = {2024-02-06},
	journal = {Physical Review A},
	author = {Bombin, H. and Martin-Delgado, M. A.},
	month = sep,
	year = {2005},
	pages = {032313},
}

@article{bennett_purification_1996,
	title = {Purification of {Noisy} {Entanglement} and {Faithful} {Teleportation} via {Noisy} {Channels}},
	volume = {76},
	issn = {0031-9007, 1079-7114},
	url = {https://link.aps.org/doi/10.1103/PhysRevLett.76.722},
	doi = {10.1103/PhysRevLett.76.722},
	language = {en},
	number = {5},
	urldate = {2024-02-06},
	journal = {Physical Review Letters},
	author = {Bennett, Charles H. and Brassard, Gilles and Popescu, Sandu and Schumacher, Benjamin and Smolin, John A. and Wootters, William K.},
	month = jan,
	year = {1996},
	pages = {722--725},
}

@article{vollbrecht_efficient_2003,
	title = {Efficient distillation beyond qubits},
	volume = {67},
	issn = {1050-2947, 1094-1622},
	url = {https://link.aps.org/doi/10.1103/PhysRevA.67.012303},
	doi = {10.1103/PhysRevA.67.012303},
	language = {en},
	number = {1},
	urldate = {2024-02-06},
	journal = {Physical Review A},
	author = {Vollbrecht, Karl Gerd H. and Wolf, Michael M.},
	month = jan,
	year = {2003},
	pages = {012303},
}

@article{miguel-ramiro_efficient_2018,
	title = {Efficient entanglement purification protocols for d -level systems},
	volume = {98},
	issn = {2469-9926, 2469-9934},
	url = {https://link.aps.org/doi/10.1103/PhysRevA.98.042309},
	doi = {10.1103/PhysRevA.98.042309},
	language = {en},
	number = {4},
	urldate = {2024-02-06},
	journal = {Physical Review A},
	author = {Miguel-Ramiro, J. and Dür, W.},
	month = oct,
	year = {2018},
	pages = {042309},
}

@article{dehaene_local_2003,
	title = {Local permutations of products of {Bell} states and entanglement distillation},
	volume = {67},
	issn = {1050-2947, 1094-1622},
	url = {https://link.aps.org/doi/10.1103/PhysRevA.67.022310},
	doi = {10.1103/PhysRevA.67.022310},
	language = {en},
	number = {2},
	urldate = {2024-02-06},
	journal = {Physical Review A},
	author = {Dehaene, Jeroen and Van Den Nest, Maarten and De Moor, Bart and Verstraete, Frank},
	month = feb,
	year = {2003},
	pages = {022310},
}

@techreport{knill_non-binary_1996,
	title = {Non-binary unitary error bases and quantum codes},
	url = {http://www.osti.gov/servlets/purl/373768-BfsVyz/webviewable/},
	language = {en},
	number = {LA-UR--96-2717, 373768},
	urldate = {2024-02-06},
	author = {Knill, E.},
	month = jun,
	year = {1996},
	doi = {10.2172/373768},
	pages = {LA--UR--96--2717, 373768},
}

@article{dur_entanglement_2003,
	title = {Entanglement {Purification} for {Quantum} {Computation}},
	volume = {90},
	issn = {0031-9007, 1079-7114},
	url = {https://link.aps.org/doi/10.1103/PhysRevLett.90.067901},
	doi = {10.1103/PhysRevLett.90.067901},
	language = {en},
	number = {6},
	urldate = {2024-02-06},
	journal = {Physical Review Letters},
	author = {Dür, W. and Briegel, H.-J.},
	month = feb,
	year = {2003},
	pages = {067901},
}

@article{bennett_teleporting_1993,
	title = {Teleporting an unknown quantum state via dual classical and {Einstein}-{Podolsky}-{Rosen} channels},
	volume = {70},
	url = {https://link.aps.org/doi/10.1103/PhysRevLett.70.1895},
	doi = {10.1103/PhysRevLett.70.1895},
	number = {13},
	urldate = {2024-02-11},
	journal = {Physical Review Letters},
	author = {Bennett, Charles H. and Brassard, Gilles and Crépeau, Claude and Jozsa, Richard and Peres, Asher and Wootters, William K.},
	month = mar,
	year = {1993},
	note = {Publisher: American Physical Society},
	pages = {1895--1899},
}

@article{popp_comparing_2023,
	title = {Comparing bound entanglement of bell diagonal pairs of qutrits and ququarts},
	volume = {13},
	copyright = {2023 The Author(s)},
	issn = {2045-2322},
	url = {https://www.nature.com/articles/s41598-023-29211-w},
	doi = {10.1038/s41598-023-29211-w},
	language = {en},
	number = {1},
	urldate = {2024-02-06},
	journal = {Scientific Reports},
	author = {Popp, Christopher and Hiesmayr, Beatrix C.},
	month = feb,
	year = {2023},
	note = {Number: 1
Publisher: Nature Publishing Group},
	pages = {2037},
}

@article{popp_special_2024,
	title = {Special features of the {Weyl}–{Heisenberg} {Bell} basis imply unusual entanglement structure of {Bell}-diagonal states},
	volume = {26},
	issn = {1367-2630},
	url = {https://dx.doi.org/10.1088/1367-2630/ad1d0e},
	doi = {10.1088/1367-2630/ad1d0e},
	language = {en},
	number = {1},
	urldate = {2024-02-06},
	journal = {New Journal of Physics},
	author = {Popp, Christopher and Hiesmayr, Beatrix C.},
	month = jan,
	year = {2024},
	note = {Publisher: IOP Publishing},
	pages = {013039},
}

@article{bennett_communication_1992,
	title = {Communication via one- and two-particle operators on {Einstein}-{Podolsky}-{Rosen} states},
	volume = {69},
	url = {https://link.aps.org/doi/10.1103/PhysRevLett.69.2881},
	doi = {10.1103/PhysRevLett.69.2881},
	number = {20},
	urldate = {2024-02-11},
	journal = {Physical Review Letters},
	author = {Bennett, Charles H. and Wiesner, Stephen J.},
	month = nov,
	year = {1992},
	note = {Publisher: American Physical Society},
	pages = {2881--2884},
}

@article{baumgartner_special_2007,
	title = {A special simplex in the state space for entangled qudits},
	volume = {40},
	issn = {1751-8113, 1751-8121},
	url = {http://arxiv.org/abs/quant-ph/0610100},
	doi = {10.1088/1751-8113/40/28/S03},
	language = {en},
	number = {28},
	urldate = {2024-02-06},
	journal = {Journal of Physics A: Mathematical and Theoretical},
	author = {Baumgartner, Bernhard and Hiesmayr, Beatrix and Narnhofer, Heide},
	month = jul,
	year = {2007},
	note = {arXiv:quant-ph/0610100},
	pages = {7919--7938},
}
\section*{Acknowledgments}
C.P. and B.C.H. acknowledge gratefully that this research was funded in whole, or in part, by the Austrian Science Fund (FWF) project P36102-N (Grant DOI:  10.55776/P36102). For the purpose of open access, the author has applied a CC BY public copyright license to any Author Accepted Manuscript version arising from this submission. The funder played no role in study design, data collection, analysis and interpretation of data, or the writing of this manuscript.

\section*{Author Contributions Statement}
C.P. developed the methods, carried out the analytical and numerical analyses, implemented the software and edited the manuscript. B.C.H. and T.C.S. revised the analyses and proposed improvements. 

\section*{Competing Interests}
 All authors declare no financial or non-financial competing interests.

\section*{Data availability statement}
All analyzed datasets were generated during the current study and are available from the corresponding author on reasonable request.

\section*{Code availability statement}
The software used to generate the reported results is published as a repository and registered open source package ``BellDiagonalQudits.jl'' \cite{popp_belldiagonalqudits_2023} available at \\
\url{https://github.com/kungfugo/BellDiagonalQudits.jl}.

 \section*{Additional information}
Correspondence and requests for materials should be addressed to C.P..
\appendix
\section{Example: Stabilizer Objects and Action Operators}
This appendix aims to provide a better understanding of the concepts introduced in Sections \ref{sec:stab-dist-prot} and \ref{sec:stabCosedDist} by giving a specific example.\\ 
\ \\
Let $d=3$ and $N=2$, $\w = e^{2\pi i/3}$. In this case, the group of Weyl errors reads:
\begin{flalign*}
    \E_2 = \lbrace W(e) ~|~ e \in \Z_3^2 \times \Z_3^2 \rbrace  
    = \lbrace W_{k_1,l_1} \otimes  W_{k_2,l_2} ~|~ k_1,l_1,k_2,l_2 \in \Z_3 \rbrace,
\end{flalign*}
where we write the error elements $e \in \Z_3^2\times\Z_3^2$ in the form 
$e = \left[ \begin{pmatrix} \begin{smallmatrix} k_1 \\  l_1 \end{smallmatrix} \end{pmatrix}, \begin{pmatrix} \begin{smallmatrix} k_2 \\ l_2  \end{smallmatrix} \end{pmatrix} \right]$.\\
\ \\
Consider the following stabilizer $S \subset \E_2$ generated by the generator $W(g) = W_{1,0}\otimes W_{1,0}$ and  the corresponding subgroup of error elements $G_S \subset \Z_3^2\times\Z_3^2$ with generating element $g= \left[ \begin{pmatrix} \begin{smallmatrix} 1 \\  0 \end{smallmatrix} \end{pmatrix}, \begin{pmatrix} \begin{smallmatrix} 1 \\ 0  \end{smallmatrix} \end{pmatrix} \right]$: 
\begin{flalign*}
    S &= \lbrace W(0), W(g), W(2g) \rbrace 
    = \lbrace \mathbbm{1}_3 \otimes \mathbbm{1}_3, W_{1,0} \otimes W_{1,0}, W_{2,0}\otimes W_{2,0} \rbrace,\\
    G_S &= \lbrace 0,g,2g \rbrace = \left\{
    \left[ \begin{pmatrix} \begin{smallmatrix} 0 \\  0 \end{smallmatrix} \end{pmatrix}, \begin{pmatrix} \begin{smallmatrix} 0 \\ 0  \end{smallmatrix} \end{pmatrix} \right],
    \left[ \begin{pmatrix} \begin{smallmatrix} 1 \\  0 \end{smallmatrix} \end{pmatrix}, \begin{pmatrix} \begin{smallmatrix} 1 \\ 0  \end{smallmatrix} \end{pmatrix} \right],
    \left[ \begin{pmatrix} \begin{smallmatrix} 2 \\  0 \end{smallmatrix} \end{pmatrix}, \begin{pmatrix} \begin{smallmatrix} 2 \\ 0  \end{smallmatrix} \end{pmatrix} \right] \right\},
\end{flalign*} 
Since the stabilizer is generated by one generator, we have $p=1$. \\
\ \\
The generator implies the decomposition of errors into cosets $C(e) = e + G_S$ and according to the symplectic product $\langle g,e \rangle$, reading
$
    \E(s) = \left\{ e =  \left[ \begin{pmatrix} \begin{smallmatrix} k_1 \\  l_1 \end{smallmatrix} \end{pmatrix}, \begin{pmatrix} \begin{smallmatrix} k_2 \\ l_2  \end{smallmatrix} \end{pmatrix} \right] ~| \langle g,e\rangle =-l_1 - l_2 = s  \right\}
$.\\ 
\ \\
The eigenvalues of $W(g)$ are $\{\w^x ~|~ x \in \Z_3 \}$, each threefold degenerated. The corresponding codespaces of $\hs_A^{\otimes 2}$, $\mQ(x) = \{ |\phi\rangle \in \hs_A^{\otimes 2} ~|~ W(g)~|\phi \rangle = w^x |\phi\rangle\}$, have dimension $d^{N-p} = d$.
We can define an encoding for this stabilizer by the mapping
\begin{flalign*}
    U := | x \rangle \otimes | k \rangle \mapsto |u_{x,k} \rangle := |k\rangle \otimes |x-k\rangle,
\end{flalign*}
defining an orthonormal basis of eigenstates of $W(g)$, i.e., codewords with $|u_{x,k}\rangle \in \mQ(x)~\forall x,k \in \Z_3$.
Since for this choice of stabilizer $S$, the stabilizer with complex conjugated elements $S^\star$ is identical to $S$, also the codespaces of $\hs_B^{\otimes 2}$ and corresponding encoding/codewords can be defined in this form. \\
\ \\
We proceed by determining the effective error action operators $T_{x}^{U,e}$. Let 
$e = \left[ \begin{pmatrix} \begin{smallmatrix} k_1 \\  l_1 \end{smallmatrix} \end{pmatrix}, \begin{pmatrix} \begin{smallmatrix} k_2 \\ l_2  \end{smallmatrix} \end{pmatrix} \right]$. Consider
\begin{flalign*}
    U^\dagger ~W(e)~ U~ (|b\rangle\otimes|j\rangle) 
    &= U^\dagger ~W_{k_1,l_1}\otimes W_{k_2,l_2}~ (|j\rangle \otimes |b-j\rangle)
    = U^\dagger ~ \w^{k_1 (j-l_1)} \w^{k_2 (b-j-l_2)} ~ |j-l_1 \rangle \otimes |b-j-l_2\rangle \\
    &= \w^{-k_1 l_1 - k_2 l_2} \w^{k_2 b} \w^{j (k_1 - k_2)}~ |b + (-l_1 - l_2) \rangle \otimes |j-l_1\rangle \\
    &= \w^{k_2 (b-l_1-l_2)}~|b+s\rangle \otimes W_{k_1-k_2, l_1}|j\rangle,
\end{flalign*}
where we used that $s = \langle g, e\rangle = -l_1 - l_2$. This shows $U^\dagger W(e) U = \sum_{x \in \Z_3} |x+s\rangle\langle x | \otimes T_{x+s}^{U,e}$ with 
\begin{flalign*}
    T_{x+s}^{U,e} = \w^{k_2 (x-l_1-l_2)} ~W_{k_1-k_2,l_1} \propto W_{k_1-k_2,l_1}.
\end{flalign*}
Note that in this case, the action operators are Weyl operators. This is due to the fact, that the chosen encoding $U$ is the canonical encoding.

\newpage
\section{Example: FIMAX routine}
Using the example of Appendix A, with $d=3$, $N=2$ and $g= \left[ \begin{pmatrix} \begin{smallmatrix} 1 \\  0 \end{smallmatrix} \end{pmatrix}, \begin{pmatrix} \begin{smallmatrix} 1 \\ 0  \end{smallmatrix} \end{pmatrix} \right]$, we demonstrate the application of FIMAX for one iteration.\\ 
\ \\
Let $\rho_{in}$ be a Bell-diagonal input state $\rho_{in} = \sum_{k,l \in \Z_3} p_{k,l} ~|\Omega_{k,l}\rangle\langle \Omega_{k,l} |$, which can be represented in the Bell basis $\{|\Omega_{0,0} \rangle, |\Omega_{1,0} \rangle, \cdots, |\Omega_{1,2} \rangle,|\Omega_{2,2} \rangle \}$ by its mixing probabilities $p_{k,l}$.
In the case of two copies of $\rho_{in}$, the induced probability distribution for two copy errors $e = \left[ \begin{pmatrix} \begin{smallmatrix} k_1 \\  l_1 \end{smallmatrix} \end{pmatrix}, \begin{pmatrix} \begin{smallmatrix} k_2 \\ l_2  \end{smallmatrix} \end{pmatrix} \right]\in \E_2$ is given by $\prob(e) = p_{k_1,l_1}p_{k_2,l_2}$. 
Consider the Bell-diagonal input state given by the following mixing probabilities (fidelities):
\begin{flalign*}
    ( p_{k,l} )_{k,l \in \Z_3} = (0.06,0.06,0.06,0.06,0.06,\mathbf{0.56},0.06,0.06,0.06)
\end{flalign*}
and the coset 
\begin{flalign*}
    C_{max} = C(\left[ \begin{pmatrix} \begin{smallmatrix} 0 \\  1 \end{smallmatrix} \end{pmatrix}, \begin{pmatrix} \begin{smallmatrix} 0 \\ 1  \end{smallmatrix} \end{pmatrix} \right]) 
    =\left\{
    \left[ \begin{pmatrix} \begin{smallmatrix} 0 \\  1 \end{smallmatrix} \end{pmatrix}, \begin{pmatrix} \begin{smallmatrix} 0 \\ 1  \end{smallmatrix} \end{pmatrix} \right],  
    \left[ \begin{pmatrix} \begin{smallmatrix} 1 \\  1 \end{smallmatrix} \end{pmatrix}, \begin{pmatrix} \begin{smallmatrix} 1 \\ 1  \end{smallmatrix} \end{pmatrix} \right] , 
    \left[ \begin{pmatrix} \begin{smallmatrix} 2 \\  1 \end{smallmatrix} \end{pmatrix}, \begin{pmatrix} \begin{smallmatrix} 2 \\ 1  \end{smallmatrix} \end{pmatrix} \right]
    \right\} \in \mathcal{C}(s_{max}) = \mathcal{C}(1).
\end{flalign*}
Calculating the probabilities $\Prob(\E(s_{max})) = \sum_{e \in \E(1)} \prob(e) = 0.5$ and similarly $\Prob(C_{max}) = \sum_{e \in C_{max}} \prob(e) = 0.314$, one obtains $\frac{\Prob(C_{max})}{\Prob(\E(s_{max}))} = 0.63$, which is the highest value among all stabilizers and cosets. For this stabilizer, there are two other cosets taking the value $0.13$ and the remaining six cosets take the value $0.02$. \\
\ \\
According to the example in Appendix A, for this stabilizer, the error action operators  are $T_x^{U,e} \propto W_{k_1-k_2,l_1}$ for $e = \left[ \begin{pmatrix} \begin{smallmatrix} k_1 \\  l_1 \end{smallmatrix} \end{pmatrix}, \begin{pmatrix} \begin{smallmatrix} k_2 \\ l_2  \end{smallmatrix} \end{pmatrix} \right]$. According to Lemma \ref{thm:coset_error_action}, the same holds for all errors of the same coset $C(e)$, so we write shorthand $T_{C(e)} \propto W_{k_1-k_2, l_1}$. In the case of $C_{max}= C(\left[ \begin{pmatrix} \begin{smallmatrix} 0 \\  1 \end{smallmatrix} \end{pmatrix}, \begin{pmatrix} \begin{smallmatrix} 0 \\ 1  \end{smallmatrix} \end{pmatrix} \right])$, this implies $T_{C_{max}} \propto W_{0,1}$. \\
\ \\
Identifying the stabilizer generated by $W(g)$ and the coset $C_{max}$ according to steps 1. and 2., Alice and Bob perform the stabilizer measurements in step 3., for which we assume outcomes $a$ and $b$ with $a-b=s_{max}=1$ in step 4. After application of the inverse encoding by both parties, Alice finally applies $T_{C_{max}}^\dagger = W_{0,2}$ to her second qudit. 
Following this routine, the output state is again of Bell-diagonal form, given by the mixing probabilities 
\begin{flalign*}
    (\hat{p}_{k,l})_{k,l \in \Z_3} = (\mathbf{0.63}, 0.13, 0.13, 0.02, 0.02, 0.02, 0.02, 0.02, 0.02).
\end{flalign*}
Note that the fidelity with the maximally entangled state $|\Omega_{0,0}\rangle$ is now precisely $\frac{\Prob(C_{max})}{\Prob(\E(s_{max}))} = 0.63$.

\end{document}